\documentclass[12pt]{article} 
\usepackage{amsmath,epsf,amssymb,cite,pifont,amsthm, mathrsfs,epsfig,  bbm, amsthm,   setspace, bbold}
\pdfoutput=1
\usepackage{graphicx}
\usepackage{caption}
\usepackage{subcaption}
\usepackage{hyperref}
\usepackage{pb-diagram,pb-xy}

\hypersetup
{
    colorlinks,	%
    citecolor=blue,%
    filecolor=black,%
    linkcolor=blue,%
    urlcolor=blue
}

\input xy
\xyoption{all}
\graphicspath{ {Figures/} }
\usepackage{Liz}

\usepackage{algorithm}
\usepackage[noend]{algpseudocode}

\numberwithin{equation}{section}
\theoremstyle{reeb}
\newtheorem{thm}[equation]{Theorem}
\newtheorem{lemma}[equation]{Lemma}
\newtheorem{prop}[equation]{Proposition}
\newtheorem{defn}[equation]{Definition}
\newtheorem{cor}[equation]{Corollary}

\newtheorem{rmk}[equation]{Remark}
\newtheorem{ex}[equation]{Example}
\newtheorem{notn}[equation]{Notation}
\newtheorem{obs}[equation]{Observation}

\title{Categorified Reeb Graphs}
\author{Vin de Silva, Elizabeth Munch, Amit Patel}
\date{\today}

\begin{document}
\maketitle

\begin{abstract}
 The Reeb graph is a construction which originated in Morse theory to study a real valued function defined on a topological space.  
 More recently, it has been used in various applications to study noisy data which creates a desire to define a measure of similarity between these structures. 
Here, we exploit the fact that the category of Reeb graphs is equivalent to the category of a particular class of cosheaf. Using this equivalency, we can define an `interleaving' distance between Reeb graphs which is stable under the perturbation of a function.
Along the way, we obtain a natural construction for smoothing a Reeb graph to reduce its topological complexity. The smoothed Reeb graph can be constructed in polynomial time.
\end{abstract}

\tableofcontents
\listoffigures


\section{Introduction}

\subsection{Purpose}
\label{Sect:Purpose}

The Reeb graph, originally defined in the context of Morse theory \cite{Reeb1946}, can be used to study properties of a space through the lens of a real-valued function by  providing a way to track and visualize the connected components of the space at levelsets of the function (Figure~\ref{F: Reeb}).
When an algorithm for computation was given in  \cite{Shinagawa1991}, the rediscovery of the Reeb graph by the computer graphics community immediately showed the Reeb graph to be an extremely useful tool in many applications.
These include shape comparison \cite{Hilaga2001, Escolano2013}, data skeletonization \cite{Ge2011, Chazal2013}, surface denoising \cite{Wood2004}, as well as choosing generators for homology classes \cite{Dey2013}; see \cite{Biasotti2008} for a survey.
Two main properties of this construction have made it extremely useful in the applied setting.
First, its dependence on the chosen function and not just on the space itself means different functions can be used to highlight different properties of the underlying space.
Second, it is rather quick to compute and thus can be used for very large data sets \cite{Parsa2012, Harvey2010,Doraiswamy2012}.

Much of the literature is dedicated to studying Reeb graphs in the context of Morse functions where a great deal can be said about its properties \cite{Cole-McLaughlin2003, Agarwal2004, Cohen-Steiner2009a}.
In addition, several variations on the Reeb graph have also been proposed and proven quite useful.
One such is Mapper \cite{Singh2007}, which applies the ideas of partial clustering to Reeb graphs in order to make the construction more robust to noise; this has found a great deal of success on big data sets \cite{Yao2009, Nicolau2011}.
A similar variation called the $\alpha$-Reeb graph was used in \cite{Chazal2013} to study data sets with 1-dimensional structure.
Another variation is the Extended Reeb graph \cite{Biasotti2000,Escolano2013}, which generalizes the theory to non-Morse functions.
Finally, a Reeb graph defined for a function on a simply-connected space cannot have any loops, and is called a contour tree~\cite{Kreveld1997,Carr2003,Pascucci2004,Szymczak_2011}. Any such contour tree can be replicated as the contour tree of a function on a 2-dimensional surface, thus allowing for informed exploration of otherwise hard-to-visualize high-dimensional data~\cite{Weber2007, Harvey2010a}.

Many of these applications involve data, where one should operate with at least a modicum of statistical integrity. Therefore it is important to consider not just individual Reeb graphs, but the whole `space' of Reeb graphs. In this paper we will:
\begin{itemize}
\item
define a distance function between pairs of Reeb graphs;

\item
show that this distance function is stable under perturbations of the input data;

\item
define `smoothing' operations on Reeb graphs (which reduce topological complexity).

\end{itemize}
The novelty of this paper is that we address these geometric questions using methods from category theory. Reeb graphs can be identified with a particular kind of cosheaf~\cite{Funk1995,Woolf_2009}, and these cosheaves may be compared using an interleaving distance of the kind studied in~\cite{Chazal2009b,Bubenik2014}. 
We pull back that interleaving distance to obtain a distance function between Reeb graphs.
While an efficient algorithm for computation of the interleaving distance is not yet available, one step of the process for construction yields a smoothed version of the given Reeb graph.
This arises from  a natural operation on cosheaves but has the added feature that it can be interpreted geometrically. 
We provide an explicit algorithm for constructing the smoothing of a given Reeb graph.
The question of simplifying a Reeb graph, perhaps to deal with noise, has arisen in several applications \cite{Bauer2014,Pascucci2007a,Doraiswamy2012,Ge2011}.
Our work differs from the solutions in those papers in that rather than doing local operations to collapse small loops, the smoothing operation is conducted globally and causes small modifications everywhere with the outcome that small loops are removed.

Recently, other approaches to defining a metric between Reeb graphs have been defined; one is based on the Gromov--Hausdorff distance \cite{Bauer2014} and the other is defined using combinatorial edits \cite{diFabio2014}. 
These methods perhaps appear more natural from the geometric perspective, whereas our method is more natural from the sheaf-theoretic perspective.
Our ideas have been inspired, in part, by the use of interleaving distances to compare join- and split-trees~\cite{Morozov2013}.
Finally, we point out that some of the category theory (without the cosheaves) appears in~\cite{Szymczak_2011}; and a very extensive and accessible study of cosheaves can be found in~\cite{Curry2014}.

\medskip
{\bf Antecedents.}
It is well known amongst sheaf theorists that a locally constant set-valued cosheaf over a manifold~$\mathbb{M}$ is equivalent to a covering space of~$\mathbb{M}$ via its \emph{display space}; see Funk~\cite{Funk1995}.  Robert MacPherson observed that if a set-valued cosheaf on~$\mathbb{M}$ is constructible with respect to a stratification (i.e.\ if it is locally constant on each stratum), then the cosheaf is equivalent to a stratified covering of~$\mathbb{M}$; see Treumann~\cite{Treumann09}, Woolf~\cite{Woolf_2009} and Curry~\cite{Curry2014} for details.  A stratified covering of the real line is what we call a Reeb graph.  This leads to an equivalence between the category of Reeb graphs and the category of constructible cosheaves over the real line.

Our definition of the interleaving distance between Reeb graphs is based on a very general framework for topological persistence developed by Bubenik and Scott~\cite{Bubenik2014} that was in turn inspired by the work of Chazal et al.~\cite{Chazal2009b} on algebraic persistence modules. Cosheaves are a particular kind of functor, as are generalized persistence modules, and the two ways of thinking overlap sufficiently to give us a metric on the category of Reeb graphs.

In this paper, we give a quite detailed exposition of the ideas involved. We describe the equivalence of categories of Funk~\cite{Funk1995} explicitly in the situation that we need it, since the eventual goal is to use this equivalence in computations. Finally, whereas most of our work can be thought of as a combination of existing ideas from two separate fields, the smoothing operators we define on Reeb graphs seem to be novel, and hint at a richer family of operations on cosheaves to be discovered.

\subsection{Reeb graphs and Reeb cosheaves}
\label{Sect:GARB}

Our starting point is a topological space $\X$  equipped with a continuous real-valued function $f: \X \to \R$. We call the pair $(\X,f)$ a `space fibered over~$\R$' or, more succinctly, an {\bf $\R$-space}. For reasons of convenience we will often abbreviate $(\X,f)$ simply to~$f$. The context will indicate whether we are thinking of~$f$ as a function or as an $\R$-space.

We can think of an $\R$-space as a 1-parameter family of topological spaces $f^{-1}(a)$, the levelsets of~$f$. The topology on~$\X$ gives information on how these spaces relate to each other. For instance, each levelset can be partitioned into connected components. How can we track these components as the parameter~$a$ varies? An answer is provided by the Reeb graph.

The {\bf (geometric) Reeb graph} of an $\R$-space $f$ is an $\R$-space $\overline{f}$ defined as follows. 
First, we define an equivalence relation on the domain of~$f$ by saying two points $x,x' \in \X$ are {equivalent} if they lie on the same levelset $f^{-1}(a)$ and on the same component of that levelset. Let $\X_f$ be the quotient space defined by this equivalence relation, and let $\overline f: \X_f \to \R$ be the function inherited from~$f$. This is the Reeb graph.
See, for example, Figure \ref{F: Reeb}.
\begin{figure}
\centering
    \includegraphics[scale=0.6]{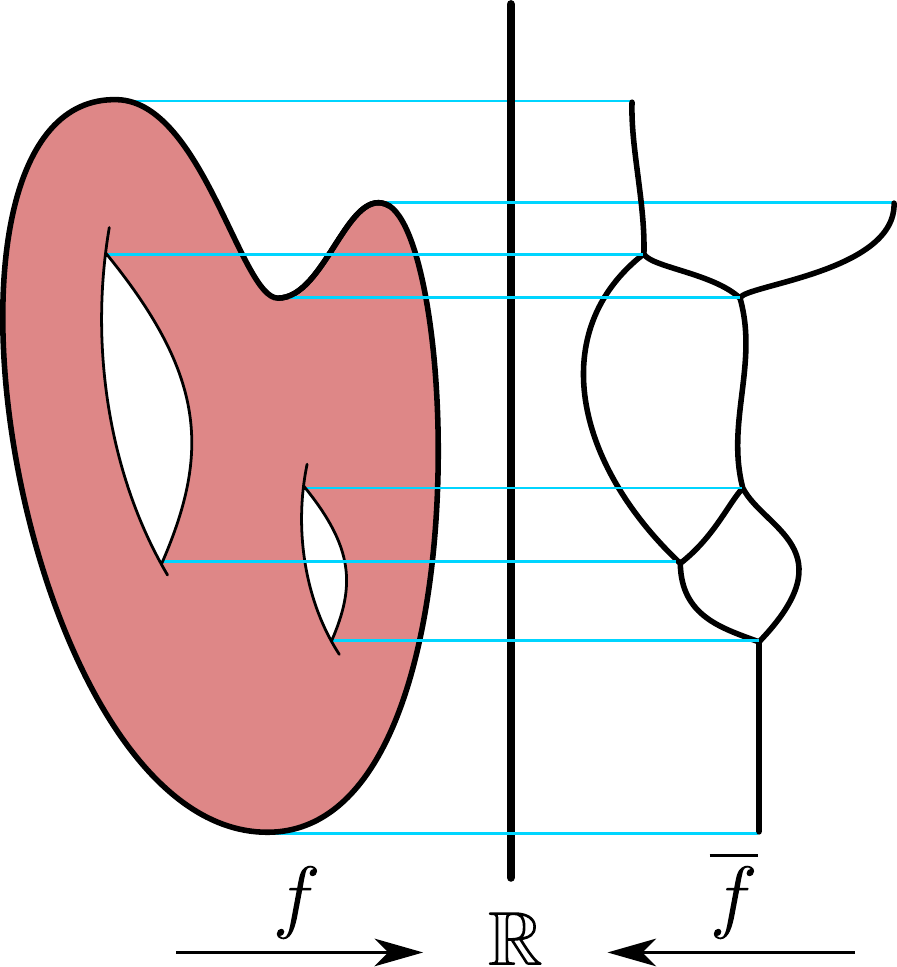}
  \caption[The Reeb graph]{The Reeb graph is used to study connected components of levelsets.}
  \label{F: Reeb}
  \vspace{-10pt}
\end{figure}

If $f$ is a Morse function on a compact manifold, or a piecewise linear function on a compact polyhedron, then its Reeb graph is topologically a finite graph with vertices at each critical value of~$f$. This situation is well studied. These examples are included in a larger class, the \emph{constructible} $\R$-spaces, which have similar good behavior. We will say more about this in Section~\ref{sec:geometric}.
If we work in greater generality, the quotient $\X \to \X_f$ can be badly behaved. Among other things, we would need to pay attention to the distinction between connected components and path components. This is not an issue for constructible $\R$-spaces, where the two concepts lead to the same outcome.

We now indicate an alternate way of recording the information stored in the geometric Reeb graph. The {\bf abstract Reeb graph} or {\bf Reeb cosheaf} of an $\R$-space~$f$ is defined to be the following collection of data
(see Figure~\ref{Fig:CosheafCondition}):
\begin{itemize}
\item
for each open interval $I \subseteq \R$, let $\Ffunc(I)$ be the set of path-components of $f^{-1}(I)$;

\item
for $I \subseteq J$, let $\Ffunc[I \subseteq J]$ be the map $\Ffunc(I) \to \Ffunc(J)$ induced by the inclusion $f^{-1}(I) \subseteq f^{-1}(J)$.

\end{itemize}
Let $\Ffunc$ denote the entirety of this data. It is easily confirmed that $\Ffunc$ is a functor (see Section~\ref{Sect:BackgroundCategory}) from the category of open intervals to the category of sets. As such, $\Ffunc$ is sometimes called a \emph{pre-cosheaf} on the real line in the category of sets.
The important point is that this information, in the constructible case, is enough to recover the geometric Reeb graph;
see Figure~\ref{Fig:CosheafConditionExample}.
\begin{figure}
        \centering
        \includegraphics[scale = 1]{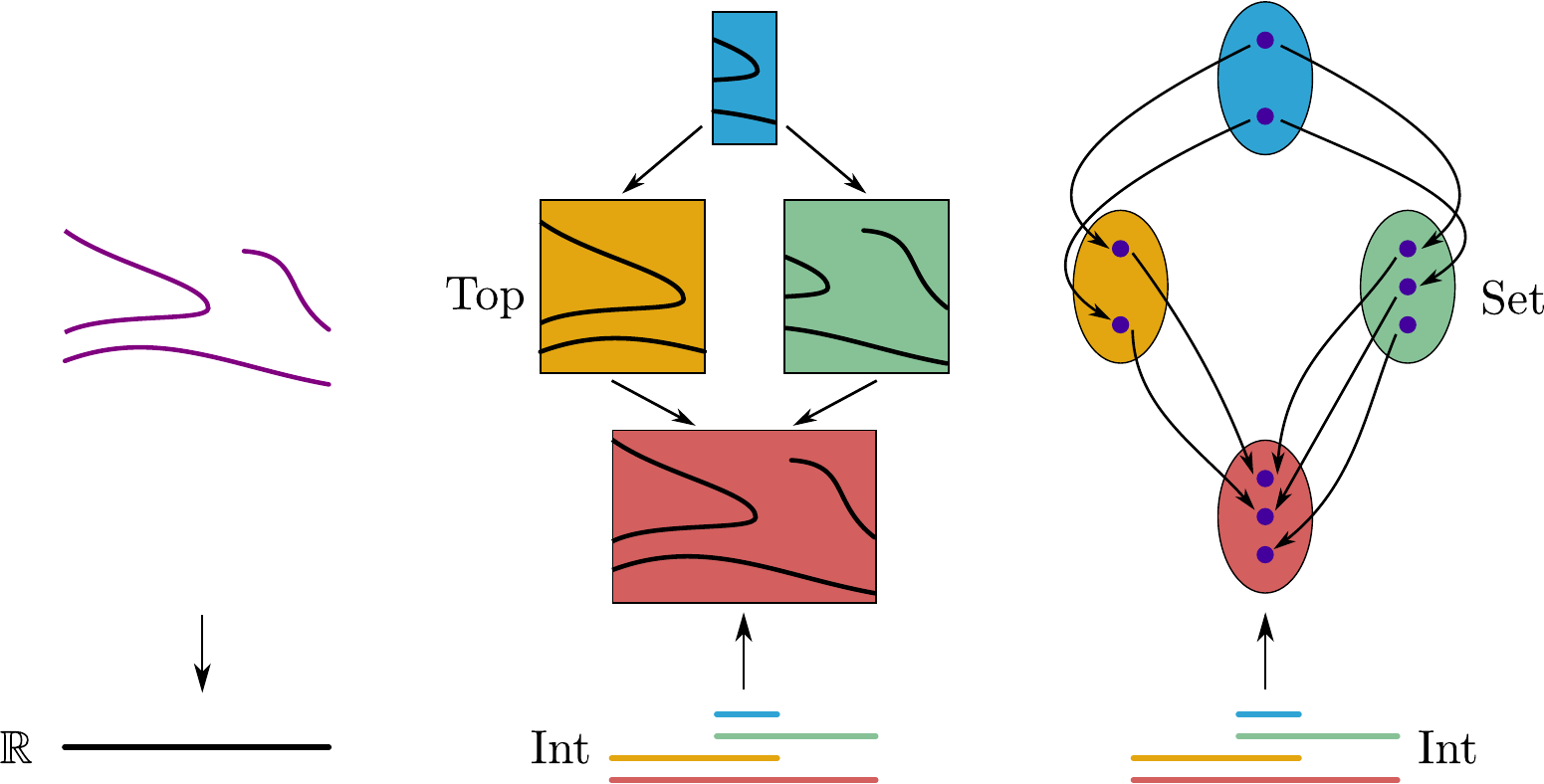}
	\caption[The Reeb cosheaf]{The geometric structure of a Reeb graph (left) may be represented categorically as a Reeb cosheaf (right). To each open interval we associate the set of connected components over that interval; to each inclusion of intervals we associate the map defined by component inclusions.}
	\label{Fig:CosheafCondition}
\end{figure}
The other important point is that it is sometimes easier to work with the pre-cosheaf than with the geometric Reeb graph.

There is considerable redundancy in the information stored by the sets $\Ffunc(I)$ and functions $\Ffunc[I \subseteq J]$ in the abstract Reeb graph of a constructible $\R$-space. For example, the components over an interval $I \cup J$ can be determined by considering the components over $I$ and~$J$ and how they are related through the components over $I \cap J$. There are similar redundancies for every cover of an interval by other intervals. When systematized, these redundancies take on a standard form: they are precisely the conditions that ensure that the pre-cosheaf is a \emph{cosheaf}.
Thus, the abstract Reeb graph is renamed the \emph{Reeb cosheaf}.

We explain these standard ideas from sheaf theory more formally in Section~\ref{sec:cosheaf}. First we recall a few concepts from category theory, which provides the language for discussing these matters.

\begin{figure}
        \centering
        \includegraphics[scale = 1]{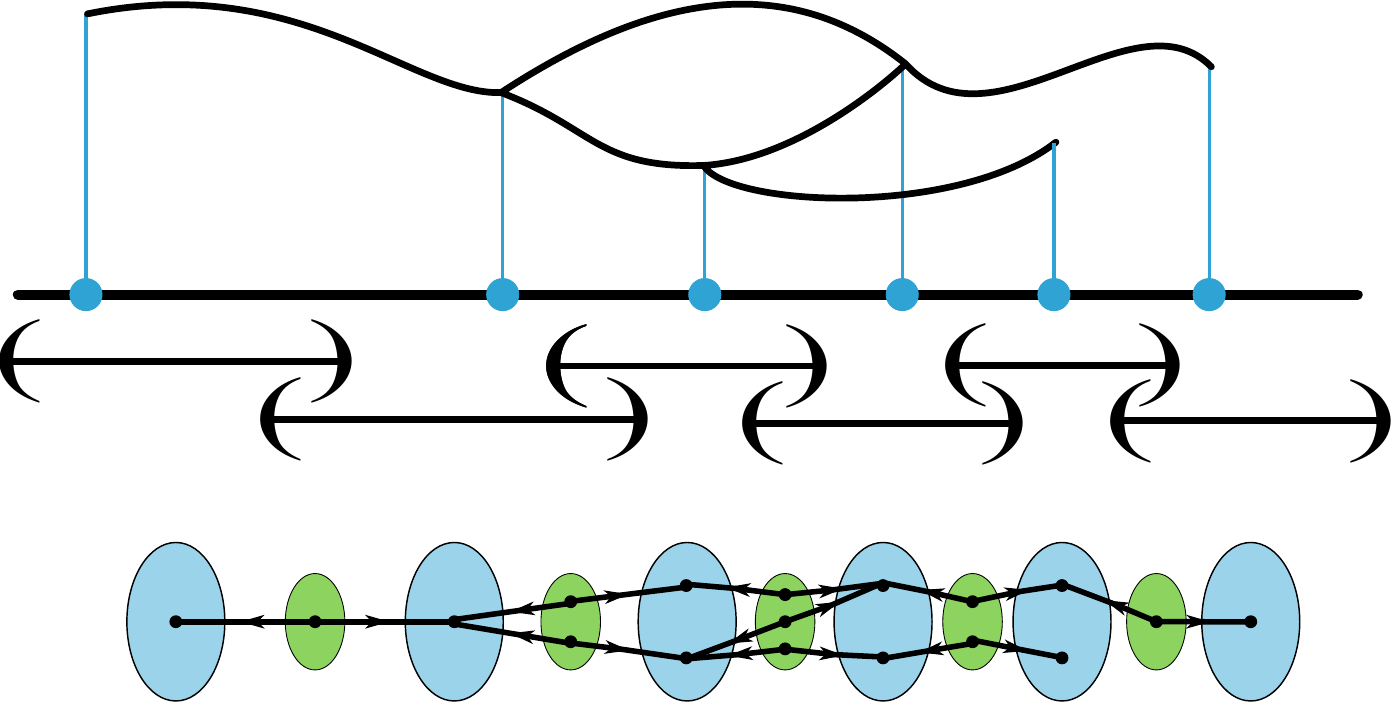}
	\caption[Recovering the Reeb graph from the Reeb cosheaf]{The Reeb cosheaf carries the same information as the Reeb graph. Here is one way the graph may be retrieved from the cosheaf. In the example, six carefully chosen intervals and their five pairwise intersections give rise to eleven sets (of components) and ten maps (of component inclusions). The graph built from this data---with eleven sets of vertices and ten sets of edges---has the same structure as the original Reeb graph.}
	\label{Fig:CosheafConditionExample}
\end{figure}

\subsection{Category theory}
\label{Sect:BackgroundCategory}

We summarize what we need from category theory. For a general reference, see~\cite{MacLane_1998}. 

A category $\CatA$ is a collection of objects $A \in \CatA$, a collection of morphisms or arrows $f : A \to A'$ between objects, and a composition operator that takes any two morphisms $f : A \to A'$ and $g : A' \to A''$ to a third morphism $gf : A \to A''$. The composition operator is associative and there is an identity morphism $\mathbb{1}_A: A \to A$ at each object~$A$. There are many examples of categories found in all branches of mathematics. Here are some common examples:

 \begin{center}
\begin{tabular}{cll}
Category & Objects & Morphisms\\\hline
$\Set$ & Sets 		     & Functions\\
$\Vect$ & Vector spaces     & Linear maps\\
$\Top$ & Topological spaces & Continuous maps\\
 \end{tabular}
 \end{center}
These are \emph{large} categories, where the collection of objects is not a set but a proper class.

\begin{ex}
\label{ex:poset}
Any partially ordered set $(P, \leq)$ can be thought of as a category $\CatP$. The objects are the elements of~$P$, and there is one morphism $p \to q$ whenever $p \leq q$ and no morphism otherwise. This is a \emph{small} category, where the collection of objects is a set.
\end{ex}

A \emph{functor} $\Ffunc : \CatA \to \CatB$ is a map between two categories. It takes each object~$A \in \CatA$ to an object $\Ffunc(A) \in\CatB$, and each morphism~$f: A \to A'$ of~$\CatA$ to a morphism $\Ffunc[f]: \Ffunc(A) \to \Ffunc(A')$ of~$\CatB$, preserving composition and identities.
%
%
A special case is the \emph{identity functor} $\onefunc_\CatA: \CatA \to \CatA$ which takes each object and morphism to itself.

A \emph{natural transformation} $\eta : \Ffunc \tO \Gfunc$ is a map between two functors $\Ffunc, \Gfunc : \CatA \to \CatB$. It consists of a collection of morphisms $\eta_A : \Ffunc (A) \to \Gfunc (A)$, one for each object $A \in \CatA$, such that for each morphism $f : A \to A'$ in $\CatA$, the following diagram commutes:	
\begin{equation}
	\label{Eqn:NatTrans}
 	\xymatrix{
 	\Ffunc (A) \ar[r]^{\Ffunc [f]} \ar[d]_{\eta_A}
	&
	\Ffunc (A') \ar[d]^{\eta_{A'}}
	\\
 	\Gfunc (A) \ar[r]_{\Gfunc [f]}
	&
	\Gfunc (A').
 	}
\end{equation}
For any functor~$\Ffunc$, there is an \emph{identity natural transformation} $\mathbb{1}_\Ffunc: \Ffunc \tO \Ffunc$, defined at each object by $(\mathbb{1}_\Ffunc)_A = \mathbb{1}_{\Ffunc(A)}$.
Natural transformations can be composed in many different ways. In particular, if $\eta: \Ffunc \tO \Gfunc$ and $\theta: \Gfunc \tO \Hfunc$ then there is a composite $\theta\eta: \Ffunc \tO \Hfunc$ defined at each object by $(\theta\eta)_A = \theta_A \eta_A$. These observations lead to the next example.

\begin{ex}
\label{ex:functor-cat}
Let $\CatA, \CatB$ be categories and suppose that $\CatA$ is small. Then the functors $\CatA \to \CatB$ themselves form a category $\CatB^\CatA$, with natural transformations as the morphisms.
\end{ex}

A natural transformation $\eta : \Ffunc \tO \Gfunc$ is a \emph{natural isomorphism} if each $\eta_A$ is an isomorphism. By defining $(\eta^{-1})_A = (\eta_A)^{-1}$ we obtain the inverse natural transformation~$\eta^{-1}$, which satisfies $\eta^{-1}\eta = \mathbb{1}_\Ffunc$ and $\eta\eta^{-1} = \mathbb{1}_\Gfunc$.
Thus, natural isomorphisms are precisely the invertible morphisms in the functor category.

\begin{rmk}[font convention]
Some of the categories in this paper---specifically, $\Pre$, $\Csh$ and $\Cshc$---are categories of functors. The objects of these categories will be written in sans-serif style: $\Ffunc, \Gfunc$. We think of these as `small' functors, the font style reminding us that we sometimes regard them as objects in a functor category.
We contrast these with various `large' functors that are defined between the major categories of interest. These we write in calligraphic style: $\FF, \GG$.
\end{rmk}

It is often convenient to have more than one equivalent categorical interpretation of a given idea.
\begin{itemize}
\item
Two functors $\FF, \GG$ are `essentially the same' if there is a natural isomorphism between them, and we write $\FF \simeq \GG$. Very often the isomorphism is canonically specified. This is much more common than the two functors being exactly equal to each other, which would be written $\FF = \GG$.

\item
Two categories $\CatA, \CatB$ are `essentially the same' if they are \emph{equivalent}. This means that there is a pair of functors $\FF : \CatA \to \CatB$ and $\GG : \CatB \to \CatA$ and a pair of natural isomorphisms $\mu : \GG \FF \tO \onefunc_{\CatA}$ and $\nu : \FF \GG \tO \onefunc_{\CatB}$.
%
%
An \emph{equivalence of categories} refers to either the complete data $(\FF, \GG, \mu, \nu)$ or one of the functors $\FF,\GG$ by itself.
\end{itemize}
Here we are mostly thinking of `large' functors, as the font style suggests.

\subsection{Road map of categories and functors}
\label{Sect:RGCF}

To develop the relationship between geometric and abstract Reeb graphs, we make use of several categories and functors.
The reader may find it helpful to consult the road map in Figure~\ref{Fig:MasterDiagram}.
\begin{figure}
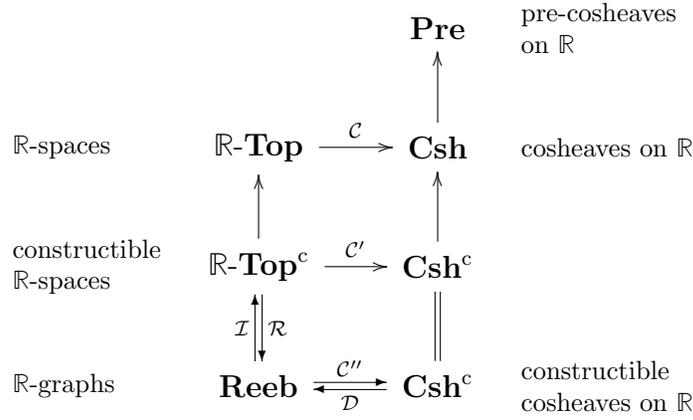

\[
\begin{diagram}
\dgARROWLENGTH=2em
\node[3]{\Pre}
\node{\text{\footnotesize \parbox{1in}{pre-cosheaves on~$\R$}}}
\\
\node{\text{\footnotesize \parbox{0.7in}{$\R$-spaces}}}
\node{\Rtop}
	\arrow{e,t}{\CC}
\node{\Csh}
	\arrow{n}
\node{\text{\footnotesize \parbox{1in}{cosheaves on~$\R$}}}
\\
\node{\text{\footnotesize \parbox{0.7in}{constructible $\R$-spaces}}}
\node{\Rtopc}
	\arrow{e,t}{\CC'}
	\arrow{n}
	\arrow{s,lr,dd}{\II}{\RR}
\node{\Cshc}
	\arrow{n}
\\
\node{\text{\footnotesize \parbox{0.7in}{$\R$-graphs}}}
\node{\Reeb}
	\arrow{e,tb,dd}{\CC''}{\DD}
\node{\Cshc}
	\arrow{n,=}
\node{\text{\footnotesize \parbox{1in}{constructible cosheaves on~$\R$}}}
\end{diagram}
\]
\caption[Road map of categories and functors]{Road map of categories and functors. Categories of geometric objects (Section~\ref{sec:geometric}) occupy the left-hand column; categories of functors (Section~\ref{sec:cosheaf}) occupy the right-hand column. The up-arrows are inclusions of categories. The bottom row is an equivalence of categories.}
\label{Fig:MasterDiagram}
\end{figure}
%
%
We define the various categories and functors in Sections \ref{sec:geometric} and~\ref{sec:cosheaf}, and establish the following relations:

\begin{itemize}
\item
$\RR\II$ is naturally isomorphic to the identity functor on $\Reeb$ (Proposition~\ref{prop:RI}).

\item
$\CC''\RR$ is naturally isomorphic to $\CC'$ (Theorem~\ref{Thm:CommutativeTriangle}).

\item
The functors $\CC'', \DD$ define an equivalence between $\Reeb$ and $\Csh^\text{c}$ (Theorem~\ref{Thm:Equivalence}).

\end{itemize}

\noindent
Subsequently, we will define a metric on $\Pre$ and smoothing operators on $\Rtop$ and $\Pre$. Through the diagram, these lead to a metric and a smoothing operator on $\Reeb$.

\section{The geometric categories}
\label{sec:geometric}

In Sections \ref{sec:rtop}, \ref{sec:rtopc} and \ref{sec:reeb} we describe the three geometric categories, from largest to smallest. In Section~\ref{sec:functor-R} we define and study the geometric Reeb functor~$\RR$.

\subsection{The category of {\bf $\R$-spaces}}
\label{sec:rtop}

An object of $\Rtop$ is a topological space~$\X$ equipped with a continuous map $f: \X \to \R$, denoted $(\X,f)$ or simply $f$. 
Point-preimages $f^{-1}(a)$ are known as \emph{levelsets} or \emph{fibers} of the $\R$-space.
A morphism $\phi : (\X, f) \to (\Y, g)$ is a continuous map ${\phi} : \X \to \Y$ such that the following diagram  commutes:
\[
	\xymatrix{
	\X \ar[rr]^{{\phi}} \ar[rd]_{f} && \Y \ar[ld]^{g}\\
	& \R &
	}
\]
Composition and identity maps are defined in the obvious way.

\begin{rmk}
Being an example of a \emph{slice category}, $\Rtop$ is sometimes named $(\Top \downarrow \R)$.
In~\cite{Szymczak_2011} it is called the \emph{category of scalar fields}.
\end{rmk}

\subsection{The category of constructible {\bf $\R$-spaces}}
\label{sec:rtopc}

We restrict to this class of spaces because the geometric Reeb graph of a general $\R$-space may be badly behaved. These spaces are compact and have finitely many `critical points' between which they have cylindrical structure.

Formally, an object of $\Rtopc$ is an $\R$-space that is isomorphic to some $(\X,f)$ constructed in the following way.
A finite set of `critical values' $S = \{a_0, a_1, \dots, a_n\}$ (listed in increasing order) is given. Then:
\begin{itemize}
\item
For $0 \leq i \leq n$, we specify a locally path-connected compact space $\V_i$.

\item
For $0 \leq i \leq n-1$, we specify a locally path-connected compact space $\E_i$.

\item
For $0 \leq i \leq n-1$, we specify continuous maps $\mathbbm{l}_i: \E_i \to \V_i$ and $\mathbbm{r}_i : \E_i \to \V_{i+1}$.
\end{itemize}
Let $\X$ be the quotient space obtained from the disjoint union of the spaces $\V_i \times \{a_i\}$ and $\E_i \times [a_i, a_{i+1}]$ by making the identifications $(\mathbbm{l}_i(x), a_i) \sim (x, a_i)$ and $(\mathbbm{r}_i(x), a_{i+1}) \sim (x, a_{i+1})$ for all $i$ and all $x \in \E_i$. Let $f: \X \to \R$ be the projection onto the second factor.

Morphisms in $\Rtopc$ are the same as morphisms in $\Rtop$ (it is a \emph{full subcategory}).

%
%
%
%

\begin{ex}
The following $\R$-spaces belong to $\Rtopc$:
(i) $\X$ is a compact differentiable manifold and $f$ is a Morse function; %
(ii) $\X$ is a compact polyhedron and $f$ is a piecewise-linear map;
(iii) $\X$ is a compact semialgebraic subset of $\R^n \times \R$ and $f$ is the projection onto the second factor.
(iii$'$) $\X$ is a compact subset of $\R^n \times \R$ definable with respect to some o-minimal structure~\cite{vdDries_1998,Coste_2000} and $f$ is the projection onto the second factor.
\end{ex}

See Figure~\ref{fig:con-R-space} for a manifold with a Morse function presented as a constructible $\R$-space.
\begin{figure}
\centerline{
\includegraphics[width = \textwidth]{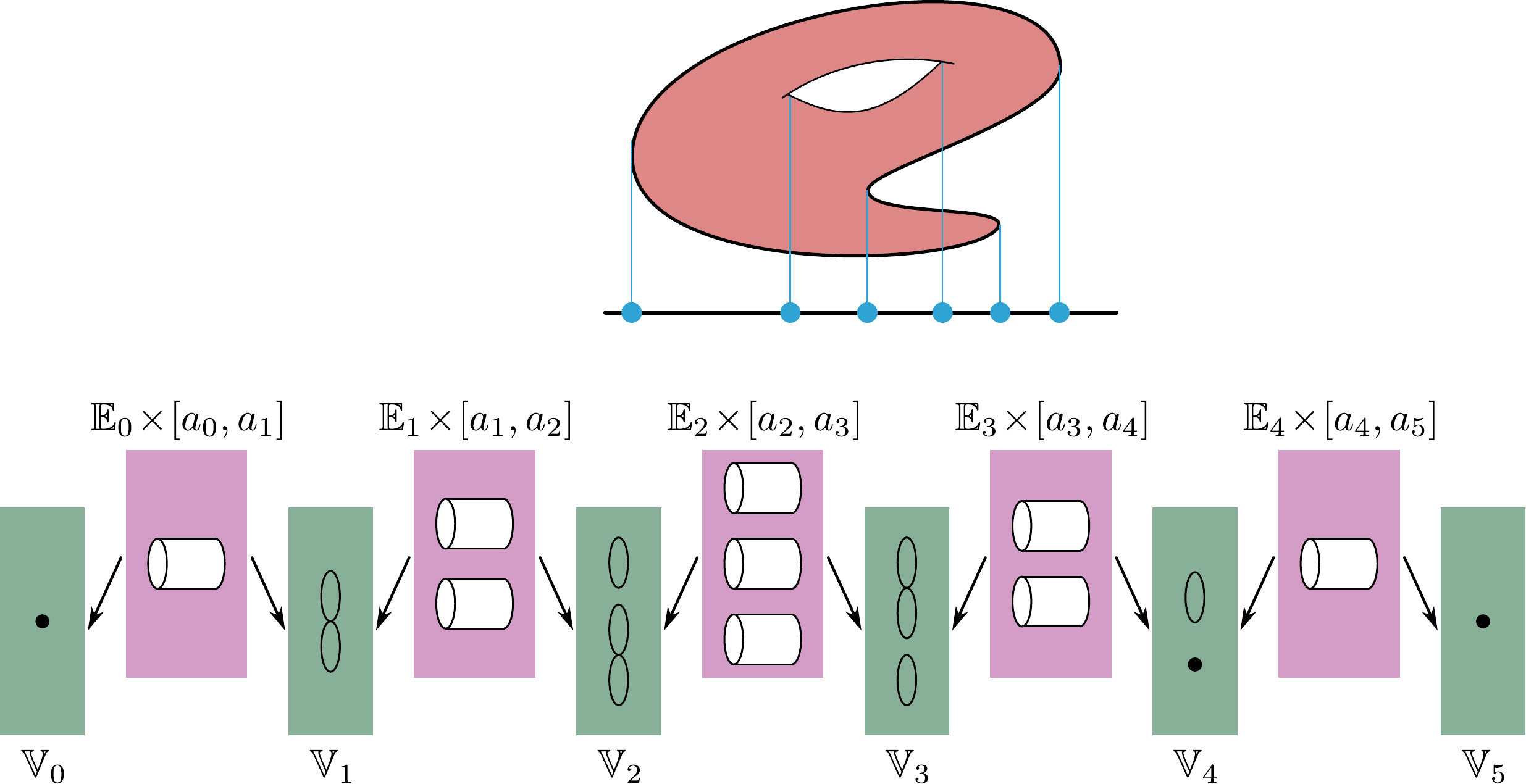}
}
\caption[Example of a constructible $\R$-space]{An example of a constructible $\R$-space. Each of the six critical values of the Morse function on this torus contributes a critical fiber $\V_i$, while the intervals between critical values contribute cylinders $\E_i \times [a_i, a_{i+1}]$ which are attached to the critical fibers using maps defined by gradient flow.
We encourage the reader to visualize the attaching maps explicitly.
}
\label{fig:con-R-space}
\end{figure}

\begin{rmk}\label{rmk:splitting}
The critical set is not uniquely specified, since we can always add extra critical points by splitting the cylinders $\E_i \times [a_i, a_{i+1}]$ appropriately. One can request a minimal critical set, but we never specifically need it.
\end{rmk}

The content of the next lemma is geometrically straightforward. We state it formally because we use it repeatedly to establish relationships between Reeb graphs and Reeb cosheaves.

\begin{lemma}[cylinder principle]
\label{lem:cylinder}
Let $(\X,f)$ be constructible with critical set $S = \{a_0, a_1, \dots, a_n\}$.
The fiber-inclusion maps
\begin{alignat*}{3}
\V_i &\longrightarrow f^{-1}(a_{i-1},a_{i+1});
\quad
&x &\mapsto (x, a_i)
\\
\E_i &\longrightarrow f^{-1}(a_{i},a_{i+1});
\quad
&x &\mapsto (x, a)
&&\qquad\text{some $a \in (a_i,a_{i+1})$}
\end{alignat*}
are homotopy equivalences which fit into diagrams 
\begin{equation}
\label{eq:transfer-1}
\begin{diagram}
\node{\E_{i-1}}
	\arrow{e,t}{\mathbb{r}_{i-1}}
	\arrow{s,l}{\simeq}
\node{\V_i}
	\arrow{s,l}{\simeq}
\node{\E_i}
	\arrow{w,t}{\mathbb{l}_i}
	\arrow{s,l}{\simeq}
\\
\node{f^{-1}(a_{i-1},a_{i})}
	\arrow{e,t}{\subseteq}
\node{f^{-1}(a_{i-1},a_{i+i})}
\node{f^{-1}(a_{i},a_{i+i})}
	\arrow{w,t}{\supseteq}
\end{diagram}
\end{equation}
that commute up to homotopy ($0 \leq i \leq n$, with $\E_0, \E_n$ interpreted as empty spaces).
The homotopy equivalences are natural, in the sense that we have commutative diagrams
\begin{equation}
\label{eq:transfer-2}
\begin{diagram}
\dgARROWLENGTH=2em
\node{\V^f_i}
	\arrow{e}
	\arrow[2]{s,l}{\alpha}
\node{f^{-1}(a_{i-1},a_{i+1})}
	\arrow[2]{s,r}{\alpha}
\node[2]{\E^f_i}
	\arrow{e}
	\arrow[2]{s,l}{\alpha}
\node{f^{-1}(a_{i},a_{i+1})}
	\arrow[2]{s,r}{\alpha}
\\
\node[3]{\text{and}}
\\
\node{\V^g_i}
	\arrow{e}
\node{g^{-1}(a_{i-1},a_{i+1})}
\node[2]{\E^g_i}
	\arrow{e}
\node{g^{-1}(a_i,a_{i+1})}
\end{diagram}
\end{equation}
whenever $\alpha: (\X,f) \to (\Y,g)$ is a morphism between $\R$-spaces with critical set~$S$. (For the left-hand maps we are identifying the spaces $\V_i, \E_i$ with the corresponding fibers at $a_i, a$ respectively.)
\end{lemma}

\begin{proof}
Thanks to the cylindrical structure between critical points, $f^{-1}(a_i,a_{i+1})$ is homotopy equivalent to any of its fibers~$\E_i$, and $f^{-1}(a_{i-1},a_{i+1})$ deformation-retracts onto its critical fiber~$\V_i$. The remaining assertions follow easily.
\end{proof}

\subsection{The category of {\bf $\R$-graphs}}
\label{sec:reeb}


An object of $\Reeb$, also known as an $\R$-graph, is a constructible $\R$-space $(\X,f)$ for which the spaces $\V_i$ and $\E_i$ are 0-dimensional (i.e.\ finite sets of points with the discrete topology).
Geometrically, it is a compact 1-dimensional polyhedron triangulated so that restriction $f|_\E$ to each edge $\E \subseteq \X$ is an embedding. 
Morphisms in $\Reeb$ are the same as morphisms in $\Rtop$.

\begin{notn}
\label{notn:Reeb-comb}
We construct an $\R$-graph $(\X,f)$ with critical set $S = \{a_0, a_1, \dots, a_n\}$ as follows:
\begin{itemize}
\item
For $0 \leq i \leq n$, we specify a finite set of vertices $V_i$, which lie over $a_i$.
\item
For $0 \leq i \leq n-1$, we specify a finite set of edges $E_{i}$ which lie over the interval $[a_i,a_{i+1}]$.

\item
For $0 \leq i \leq n-1$, we specify attaching maps $\ell_{i}: E_i \to V_i$ and $r_{i}: E_i \to V_{i+1}$.
\end{itemize}
The space~$\X$ is the quotient of the disjoint union of the spaces $V_i \times \{a_i\}$ and $E_i \times [a_i,a_{i+1}]$ with respect to the identifications $(\ell_i(e), a_i) \sim (e,a_i)$ and $(r_i(e), a_{i+1}) \sim (e, a_{i+1})$, with the map~$f$ being the projection onto the second factor.
See Figure~\ref{fig:con-R-graph}.
If we wish to add extra points to the critical set, we can retain this description by splitting the edges at new vertices over the new critical points.
\end{notn}

\begin{figure}
\centerline{
\includegraphics[width = \textwidth]{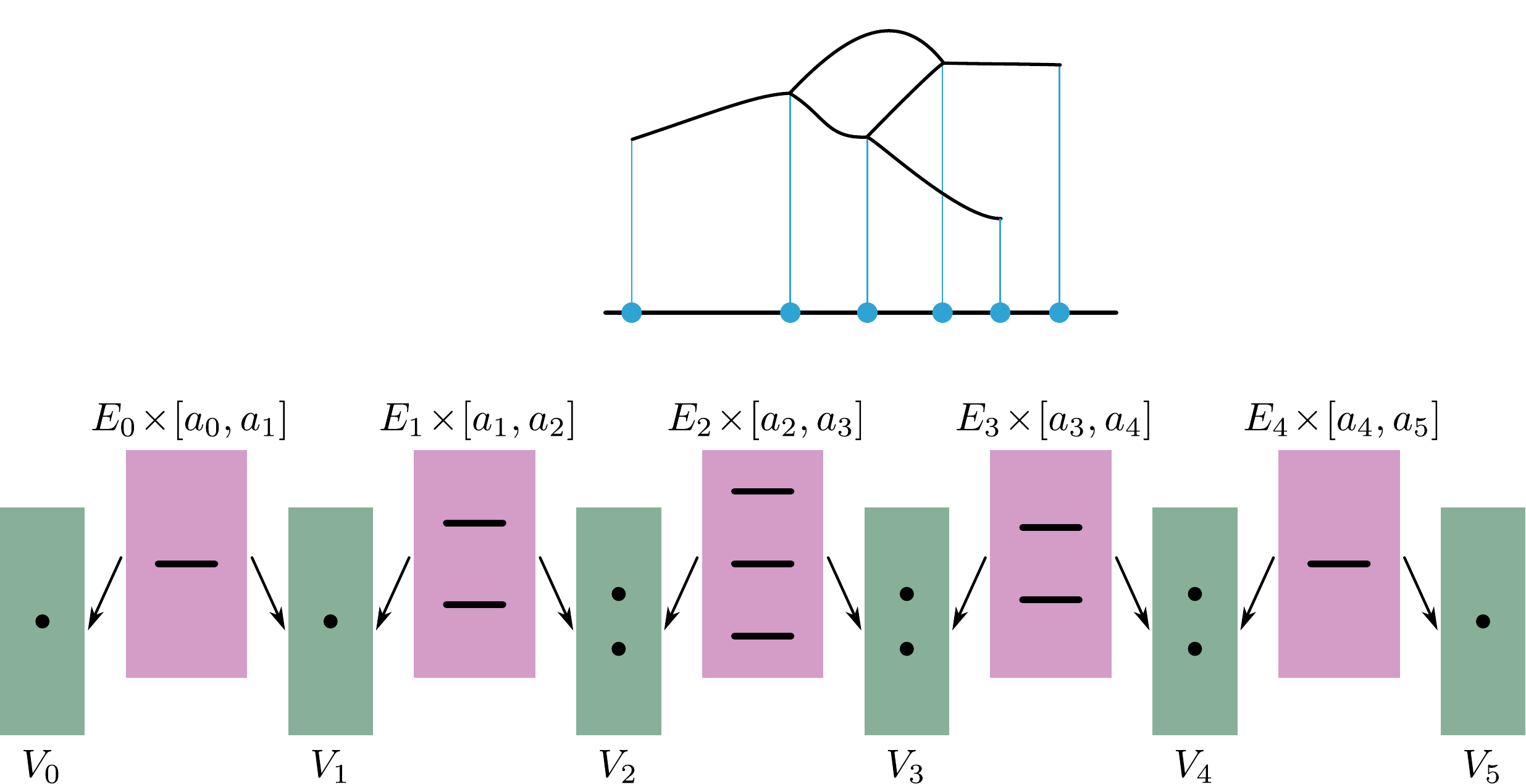}
}
\caption[Example of an $\R$-graph]{An example of an $\R$-graph with its presentation as a constructible $\R$-space. It is the Reeb graph of the $\R$-space in Figure~\ref{fig:con-R-space}, its presentation obtained by applying the $\pi_0$ functor to the presentation indicated there.
Indeed, the reader may verify that $V_i = \pi_0(\V_i)$ and $E_i = \pi_0(\E_i)$ and that the attaching maps are obtained in the manner described.}
\label{fig:con-R-graph}
\end{figure}

The restriction of $f$ over each open interval $(a_i,a_{i+1})$ is a covering map. We use this fact in the next proposition, which gives a combinatorial description of the morphisms of $\Reeb$.

\begin{prop}
\label{Prop:MapData}
Let $(\X,f)$, $(\Y,g)$ be $\R$-graphs with a common critical set~$S = \{a_0, a_1, \dots, a_n\}$ and described as above. A morphism $\phi:(\X,f) \to (\Y,g)$ is exactly specified by the following data:
\begin{itemize}
\item
Maps $\phi_i^V: V_i^f \to V_i^g$ for $0 \leq i \leq n$.

\item
Maps $\phi_i^E: E_i^f \to E_i^g$ for $0 \leq i \leq n-1$.

\item
The consistency conditions $\phi^V_i \ell_i^f = \ell_i^g \phi^E_i$ and $\phi^V_{i+1} r_i^f = r_i^g \phi^E_i$ are satisfied, for $0 \leq i \leq n-1$.
\end{itemize}
\end{prop}

\begin{proof}
Any morphism $\phi: (\X,f) \to (\Y,g)$ defines consistent vertex and edge maps as above (the covering map structure between critical points guarantees that each edge of~$\X$ maps to exactly one edge of~$\Y$, in a unique way once the edge is chosen). Conversely, a morphism can be specified by defining continuous maps on the vertices and edges of~$\X$ in a consistent way; the data above provide that.
\end{proof}

The requirement of a common critical set is no restriction, because we can take the union of critical sets for $(\X,f)$ and~$(\Y,g)$ to obtain a critical set for both. On the other hand, for computational purposes we may wish to be more economical; see Section~\ref{subsec:morphisms-diff-S}.

\subsection{The Reeb functor $\RR$}
\label{sec:functor-R}

The Reeb functor $\RR$ converts a constructible $\R$-space to an $\R$-graph, its {\bf Reeb graph}. We provisionally define it as a functor $\Rtop \to \Rtop$, then show that it restricts to a functor $\Rtopc {\to} \Reeb$.
%
%

\begin{lemma}[quotient principle {\cite[Proposition~3.8.2]{Sutherland2009}}]
\label{lem:quotient}
Let $(\X/{\sim})$ be the quotient of a topological space~$\X$ by an equivalence relation~$\sim$, and let $\Y$ be another topological space. For any continuous function $\X \to \Y$ which is constant on equivalence classes, the induced function $(\X/{\sim}) \to \Y$ is continuous with respect to the quotient topology.
\qed
\end{lemma}

Let $(\X,f)$ be an $\R$-space, with geometric Reeb graph $(\X_f, \bar{f})$.
Recall that $\X_f$ is the quotient space of~$\X$ by the relation whose equivalence classes are the path-components of the levelsets of~$f$.
It follows from the quotient principle that $\bar{f}$ is continuous, so $(\X_f, \bar{f})$ is an $\R$-space.
The quotient map $\X \to \X_f$ defines a morphism $(\X,f) \to (\X_f, \bar{f})$ in the category of $\R$-spaces; we label this morphism $\rho_f = \rho_{(\X,f)}$.

Now consider a morphism $\alpha: (\X,f) \to (\Y,g)$. Since $\alpha$ preserves levelsets and (being continuous) carries path-connected sets to path-connected sets, the composite map $\X \stackrel{\alpha}{\to} \Y \to \Y_g$ is constant on equivalence classes. By the quotient principle, the induced map $\bar\alpha: \X_f \to \Y_g$ is continuous.
Since $\bar{g} \bar{\alpha} = \bar{f}$ this defines a morphism $\bar\alpha: (\X_f, \bar{f}) \to (\Y_g, \bar{g})$.

\begin{prop}
\label{prop:Reeb-functor}
The formulas $\RR(\X, f) = \left(\X_f, \bar{f}\right)$ and $\RR[\alpha] = \bar\alpha$ define a functor $\Rtop \to \Rtop$. The collection of maps $\rho = (\rho_f)$ constitute a natural transformation $\onefunc_{\Rtop} \tO \RR$.
\end{prop}

\begin{obs}
\label{obs:Reeb-rho}
In other words $\RR$ is a \emph{pointed endofunctor} of $\Rtop$: an endofunctor that is the target of a natural transformation from the identity functor.
We call $\RR$ the \emph{Reeb functor} and its `basepoint' $\rho$ the \emph{canonical projection}.
\end{obs}

\begin{proof}[Proof of Proposition~\ref{prop:Reeb-functor}]
Notice that $\bar\alpha$ is the unique map that makes the following diagram commute:
\begin{equation*}
\xymatrix{
\X
	\ar[r]^{\alpha}
	\ar[d]_{\rho_f}
&
\Y
	\ar[d]^{\rho_g}
\\
\X_f
	\ar[r]^{\bar\alpha}
&
\Y_g
}
\end{equation*}
Uniqueness implies that $\RR[\cdot]$ respects identities and composition, so $\RR$ is a functor. Of course, these facts are easily verified directly. The commuting of the square is what makes $\rho$ a natural transformation.
\end{proof}

\begin{prop}
The Reeb functor carries constructible $\R$-spaces to $\R$-graphs.
\end{prop}

To see how this works, compare Figure~\ref{fig:con-R-space} to Figure~\ref{fig:con-R-graph}.

\begin{proof}
Given a constructible $\R$-space $(\X,f)$ it is clear how $\RR(\X,f)$ should be described as an $\R$-graph (using Notation~\ref{notn:Reeb-comb}):
\[
V_i = \pi_0 (\V_i),
\quad
E_i = \pi_0 (\E_i),
\quad
\ell_i = \pi_0 [\mathbbm{l}_i],
\quad
r_i = \pi_0 [\mathbbm{r}_i].
\]
We have a tautological bijection from $\X_f$ to this graph~$\mathbb{G}$, since the points of~$\mathbb{G}$ are precisely the path-components of the levelsets of~$f$. This map preserves levelsets. It remains to show that it is a homeomorphism.

First we show that it is continuous. Let $\hat{\X}$ denote the disjoint union of the spaces $\V_i \times \{a_i\}$ and $\E_i \times [a_i, a_{i+1}]$. Consider the composite $\hat{\X} \to \X \to \X_f \to \mathbb{G}$; the first two maps are quotient maps. This composite is continuous because the path-components of a locally path-connected space are open. Applying the quotient principle twice, it follows that $\X_f \to \mathbb{G}$ is continuous.

The proof is completed by invoking the standard result that  a continuous bijection from a compact space ($\X_f$) to a Hausdorff space ($\mathbb{G}$) is a homeomorphism~\cite[Theorem~5.9.1]{Sutherland2009}.
\end{proof}

\begin{prop}
\label{prop:RI}
Each $\R$-graph is naturally isomorphic to its Reeb graph.
\end{prop}

\begin{proof}
Each levelset of an $\R$-graph is a finite discrete space, so the equivalence classes are singletons. Thus the canonical projection~$\rho$ is a homeomorphism in these cases.
\end{proof}

Thus we can think of the Reeb functor as a projection operator $\Rtopc \to \Reeb$. Henceforth, we will mostly reserve the symbol $\RR$ for the functor with this particular domain and codomain.
Proposition~\ref{prop:RI} can be restated as the assertion that $\rho$ restricts to a natural isomorphism $\onefunc_\Reeb \tO \RR\II$.

\section{The cosheaf categories}
\label{sec:cosheaf}

We now describe the three cosheaf categories, from largest to smallest, in Sections \ref{sec:pre}, \ref{sec:csh} and~\ref{sec:cshc}.
In Section~\ref{sec:functor-C} we study the Reeb cosheaf functor~$\CC$, and in Section~\ref{sec:equivalence} we show that it defines an equivalence of categories.

The idea behind the cosheaf categories is that we can study an $\R$-space by inspecting its behavior over subintervals of~$\R$. To this end, let $\Int$ denote the category whose objects are open intervals $I \subseteq \R$ and whose morphisms are the inclusions $I \subseteq J$. (This is an instance of Example~\ref{ex:poset}.)

\subsection{The category of pre-cosheaves}
\label{sec:pre}

The largest of the three cosheaf categories is $\Pre = \Set^\Int$, the category of functors $\Int \to \Set$ with natural transformations as morphisms (Example~\ref{ex:functor-cat}). The elements of $\Pre$ are called {\bf pre-cosheaves}.

\begin{rmk}
More generally for any category $\CCC$ we can consider $\Pre(\CCC) = \CCC^\Int$, the category of pre-cosheaves in~$\CCC$ over the real line.
\end{rmk}

\begin{ex}
Let $(\X,f)$ be an $\R$-space. This determines a pre-cosheaf $\Xfunc \in \Pre(\Top)$ over the real line, as follows: for every interval $I$ we let $\Xfunc(I)$ be the topological space $f^{-1}(I)$, and for every pair $I \subseteq J$ we let $\Xfunc[I \subseteq J]$ be the inclusion map $f^{-1}(I) \subseteq f^{-1}(J)$.
%
\end{ex}

\begin{ex}
\label{ex:pre-cosheaves}
The previous example generates many others. We can post-compose $\Xfunc$ with any functor $\GG: \Top \to \CCC$ to obtain a pre-cosheaf $\GG \Xfunc \in \Pre(\CCC)$. For example:
\begin{itemize}
\item
Let $\mathrm{H}_k$ denote singular $k$-homology; then $\mathrm{H}_k\Xfunc$ is a pre-cosheaf in $\mathbf{Ab}$, abelian groups.

\item
Let $\pi_0$ denote the set of path-components of a space; then $\pi_0\Xfunc$ is a pre-cosheaf in $\Set$.

\item
Let $\overline\pi_0$ denote the set of connected components of a space; then $\overline\pi_0\Xfunc$ is a pre-cosheaf in $\Set$.

\end{itemize}
Thus $\pi_0 \Xfunc(I) = \pi_0(f^{-1}(I))$, for example.
The key requirement is that the operations $\mathrm{H}_k$, $\pi_0$, $\overline\pi_0$ be functors: they specify an object for each topological space, and a morphism for each continuous map. For instance, if $\phi: \Y \to \Z$ is a continuous map then each path-component of~$\Y$ maps into a path-component of~$\Z$. This defines $\pi_0[\phi] : \pi_0(\Y) \to \pi_0(\Z)$.
\end{ex}

\subsection{The category of cosheaves}
\label{sec:csh}

The second category in the right-hand column is $\Csh$, the category of {\bf cosheaves} in $\Set$ over the real line. 
It is a full subcategory of $\Pre$: it is defined by specifying which pre-cosheaves are cosheaves, and declaring that cosheaf morphisms are the same as pre-cosheaf morphisms.

A \emph{cosheaf} is a pre-cosheaf $\Ffunc$ which satisfies the following `gluing' property. Let $U$ be an open interval and let $(I_p \mid p \in P)$ be a family of open intervals whose union is~$U$. Then we ask that $\Ffunc(U)$ be the \emph{colimit} of the following diagram:
\begin{equation}
\label{eq:cosheaf}
\coprod_{p,q} \Ffunc(I_p \cap I_q)
\rightrightarrows
\coprod_p \Ffunc(I_p)
\end{equation}
This must be true for every $U$ and every cover $(I_p)$. In particular, this implies $\Ffunc(\emptyset) = \emptyset$ so on the left-hand side of \eqref{eq:cosheaf} we consider only the terms with $I_p \cap I_q$ nonempty.

Here are three interpretations of the gluing property.

\begin{enumerate}
\item
$\Ffunc(U)$ is obtained from the disjoint union $\coprod_p \Ffunc(I_p)$ by identifying all pairs of points
\begin{align*}
\Ffunc[I_p \cap I_q \subseteq I_p](x) &\in \Ffunc(I_p)
\\
\Ffunc[I_p \cap I_q \subseteq I_q](x) &\in \Ffunc(I_q)
\end{align*}
where $p,q$ are indices with $I_p \cap I_q$ nonempty and $x \in \Ffunc(I_p \cap I_q)$.

\item
$\Ffunc(U)$ is the set of connected components of the graph with a vertex for every element in the disjoint union $\coprod_p \Ffunc(I_p)$ and an edge for every element in the disjoint union $\coprod_{p,q} \Ffunc(I_p \cap I_q)$. The maps $\Ffunc[I_p \cap I_q \subseteq I_p]$ and $\Ffunc[I_p \cap I_q \subseteq I_q]$ indicate the vertices to which each edge is glued.

\item
$\Ffunc(U)$ is characterized by the following universal property. Let $Z$ be a set, and suppose that maps $\zeta_p: \Ffunc(I_p) \to Z$ are given for all~$p$, such that for all $p,q$ with $I_p \cap I_q$ nonempty the following two maps $\Ffunc(I_p \cap I_q) \to Z$ are equal:
\begin{equation}
\label{eq:cosheaf-3}
\zeta_p \circ \Ffunc[I_p \cap I_q \subseteq I_p]
=
\zeta_q \circ \Ffunc[I_p \cap I_q \subseteq I_q]
\end{equation}
Then there is a unique map $\zeta: \Ffunc(U) \to Z$ such that
\begin{equation}
\label{eq:cosheaf-4}
\zeta_p
=
\zeta \circ \Ffunc[I_p \subseteq U]
\end{equation}
for all~$p$.
\end{enumerate}
The first two interpretations are valid in~$\Set$. The third interpretation is meaningful in any category.
The universal property characterizes $\Ffunc(U)$, and the maps to it from the $\Ffunc(I_p)$ and the $\Ffunc(I_p \cap I_q)$, uniquely up to a canonical isomorphism.

\begin{ex}
In Example~\ref{ex:pre-cosheaves} above, $\mathrm{H}_k \Xfunc$ is not in general a cosheaf. Consider open intervals $U = I \cup J$ and $N = I \cap J$. The Mayer--Vietoris theorem gives the following exact sequence of abelian groups:
\[
\dots
\To
\mathrm{H}_k(f^{-1}(N))
\To
\mathrm{H}_k(f^{-1}(I))
\oplus
\mathrm{H}_k(f^{-1}(J))
\To
\mathrm{H}_k(f^{-1}(U))
\stackrel{\partial}{\To}
\mathrm{H}_{k-1}(f^{-1}(N))
\To
\dots
\]
If the map labeled~$\partial$ were zero, then the exactness of this sequence implies that the gluing condition holds for the cover $U = I \cup J$. When $\partial$ is not zero, the gluing condition fails for this cover.
\end{ex}

\begin{ex}
Continuing with Example~\ref{ex:pre-cosheaves}, we will later see that $\pi_0 \Xfunc$ is always a cosheaf (Proposition~\ref{prop:Cf-cosheaf}), but $\overline\pi_0 \Xfunc$ sometimes is not (Example~\ref{ex:non-cosheaf}).
\end{ex}

\subsection{The category of constructible cosheaves}
\label{sec:cshc}

The third category in the right-hand column is $\Cshc$, the category of {\bf constructible cosheaves} in $\Set$ over the real line. It is a full subcategory of the category of cosheaves, defined by specifying which cosheaves are constructible and using the same morphisms as before.

\begin{defn}
\label{defn:constructible-cosheaf}
A cosheaf or pre-cosheaf $\Ffunc$ is constructible if each $\Ffunc(I)$ is finite and there exists a finite set $S \subset \R$ of `critical values' such that:
\begin{itemize}
\item
if $I \subseteq J$ are open intervals with $I \cap S = J \cap S$ then $\Ffunc[I \subseteq J]$ is an isomorphism;

\item
if $I$ is contained in $(-\infty, \min(S))$ or $(\max(S), +\infty)$ then $\Ffunc(I)$ is empty.
\end{itemize}
As with constructible $\R$-spaces, if the conditions hold for some~$S$ then they hold for any $S' \supseteq S$. 
\end{defn}

For a given critical set $S = \{a_0, a_1, \dots, a_n\}$ the following `zigzag' diagram in~$\Int$ is of particular importance:
\begin{equation}
\begin{diagram}
\node{\makebox[2em]{$(-\infty,a_1)$}}
\node[2]{\makebox[2em]{$(a_0,a_2)$}}
\node[2]{\dots}
\node[2]{\makebox[2em]{$(a_{n-2},a_n)$}}
\node[2]{\makebox[2em]{$(a_{n-1},+\infty)$}}
\\
\node[2]{\makebox[2em]{$(a_0,a_1)$}}
	\arrow{nw}
	\arrow{ne}
\node[2]{\makebox[2em]{$(a_{1},a_{2})$}}
	\arrow{nw}
	\arrow{ne}
\node[2]{\dots}
	\arrow{ne}
\node[2]{\makebox[2em]{$(a_{n-1},a_n)$}}
	\arrow{nw}
	\arrow{ne}
\end{diagram}
\label{eq:entrance}
\end{equation}

\begin{notn}
\label{notn:cosheaf-comb}
For a constructible cosheaf~$\Ffunc$ with critical set~$S$, write
\begin{alignat*}{2}
V^\Ffunc_i &= \Ffunc((a_{i-1}, a_{i+1})),
\quad
&& \text{for $0 \leq i \leq n$}
\\
E^\Ffunc_i &= \Ffunc((a_{i}, a_{i+1}))
\quad
&& \text{for $0 \leq i \leq n-1$}
\end{alignat*}
where $a_{-1} = -\infty$ and $a_{n+1} = +\infty$; and
\begin{alignat*}{2}
\ell^\Ffunc_i & = \Ffunc[(a_{i}, a_{i+1}) \subseteq (a_{i-1}, a_{i+1})]
&&: E^\Ffunc_i \to V^\Ffunc_i
\\
r^\Ffunc_i & = \Ffunc[(a_{i}, a_{i+1}) \subseteq (a_{i}, a_{i+2})]
&&: E^\Ffunc_i \to V^\Ffunc_{i+1}
\end{alignat*}
for $0 \leq i \leq n-1$.
\end{notn}

As we will see in Proposition~\ref{prop:cshc-data}, this collection of combinatorial data suffices to define a constructible cosheaf that is unique up to canonical isomorphism.
To get to this result---and more importantly to move towards the equivalence of categories, Theorem~\ref{Thm:Equivalence}---we establish a combinatorial description of morphisms between constructible cosheaves.

\begin{prop}[combinatorial description of cosheaf morphisms]
\label{prop:cshc-morphism}
Let $\Ffunc, \Gfunc$ be constructible cosheaves, and let $S = \{a_0, a_1, \dots, a_n\}$ be a common critical set for them.
A morphism $\psi: \Ffunc \tO \Gfunc$ gives rise to the following data:
\begin{itemize}
\item
Maps $\psi_i^V: V_i^\Ffunc \to V_i^\Gfunc$ for $0 \leq i \leq n$.

\item
Maps $\psi_i^E: E_i^\Ffunc \to E_i^\Gfunc$ for $0 \leq i \leq n-1$.

\item
Consistency conditions: $\psi^V_i \ell_i^\Ffunc = \ell_i^\Gfunc \psi^E_i$ and $\psi^V_{i+1} r_i^\Ffunc = r_i^\Gfunc \psi^E_i$ for $0 \leq i \leq n-1$.
\end{itemize}
Conversely, any collection of maps $\psi^V_i, \psi^E_i$ satisfying the consistency conditions arises in this way from a unique morphism $\psi: \Ffunc \tO \Gfunc$.
\end{prop}

\begin{rmk}
As usual, the requirement of a common critical set is no restriction, because we can take the union of critical sets for $\Ffunc$ and~$\Gfunc$ to obtain a critical set for both.
\end{rmk}

\begin{proof}
The first assertion is clear: we set $\psi^V_i = \psi_{(a_{i-1},a_{i+1})}$ and $\psi^E_i = \psi_{(a_i,a_{i+1})}$, and the consistency conditions follow from the naturality of~$\psi$ with respect to the inclusions $(a_i, a_{i+1}) \subseteq (a_{i-1},a_{i+1})$ and $(a_i, a_{i+1}) \subseteq (a_{i},a_{i+2})$.

For the converse, we must show that the collection of maps $ \psi_{(a_{i-1},a_{i+1})} = \psi^V_i$ and $ \psi_{(a_{i},a_{i+1})} = \psi^E_i$ extends uniquely to a natural transformation $\psi = (\psi_I \mid I \in \Int)$.

We begin by showing that $\psi_I$ is uniquely determined for intervals of the form $I = (a_j,a_k)$. The starting data provides these maps for intervals of `length' (i.e. $k-j$ equal to) 1 or~2. Longer intervals can be expressed as a union of intervals of length two. We then have:
\begin{align*}
\Ffunc((a_j,a_k))
&=
\colim
\left[
\coprod_{i=j+1}^{k-2} {\Ffunc((a_i,a_{i+1}))}
\rightrightarrows
\coprod_{i=j}^{k-2} {\Ffunc((a_i,a_{i+2}))}
\right]
\\
\Gfunc((a_j,a_k))
&=
\colim
\left[
\coprod_{i=j+1}^{k-2} {\Gfunc((a_i,a_{i+1}))}
\rightrightarrows
\coprod_{i=j}^{k-2} {\Gfunc((a_i,a_{i+2}))}
\right]
\end{align*}
By naturality, the desired map $\psi_{(a_j,a_k)}$ must be compatible with the maps $\psi^V_i, \psi^E_i$ that map the terms in the colimit for~$\Ffunc$ to the terms in the colimit for~$\Gfunc$. 
In particular, $\psi_{(a_j,a_k)}$ must factor the maps
\[
\zeta_i = 
\Gfunc[(a_i,a_{i+2}) \subseteq (a_j,a_k)] \circ \psi_i^V
:
\Ffunc((a_i,a_{i+2})) \to \Gfunc((a_j,a_k)).
\]
The consistency conditions imply that the family $(\zeta_i)$ satisfies \eqref{eq:cosheaf-3}. From the universal property for colimits, there is a unique $\psi_{(a_j,a_k)} = \zeta$ compatible with the $\psi_i^V, \psi_i^E$.

For an inclusion $I \subseteq J$ between two such intervals, the naturality condition is that the square
\[
\begin{diagram}
\dgARROWLENGTH=1.5em
\node{\Ffunc(I)}
	\arrow[2]{e,t}{\Ffunc[I \subseteq J]}
	\arrow{s,l}{\psi_I}
\node[2]{\Ffunc(J)}
	\arrow{s,r}{\psi_J}
\\
\node{\Gfunc(I)}
	\arrow[2]{e,t}{\Gfunc[I \subseteq J]}
\node[2]{\Gfunc(J)}
\end{diagram}
\]
commutes. If $I$ is an interval of length 1 or~2, this is precisely the compatibility demanded in the construction of $\psi_J$. If $I$ is a longer interval, then the diagram commutes when $\Ffunc(I)$ is replaced by any of the the individual terms in its colimit diagram. The universal property for colimits then implies that the two sides of the square $\Ffunc(I) \to \Gfunc(J)$ are equal, being the unique map compatible with the map on each individual term.

To finish, we consider arbitrary open intervals. Any such~$I$ is contained in a unique maximal interval $\hat{I} = (a_j,a_k)$ that meets $S$ in the same subset. Since $\Gfunc[I \subseteq \hat{I}]$ is an isomorphism, we can and must define
\[
\psi_I = 
\Gfunc[I \subseteq \hat{I}]^{-1}
\circ
\psi_{\hat{I}}
\circ
\Ffunc[I \subseteq \hat{I}]
\]
to satisfy naturality for $I \subseteq \hat{I}$.
That done, naturality for $I \subseteq J$ follows from naturality for $\hat{I} \subseteq \hat{J}$.
\end{proof}

\begin{cor}
\label{cor:cshc-isomorphism}
Let $\psi: \Ffunc \tO \Gfunc$ be a morphism between constructible cosheaves with common critical set $S = \{a_0, a_1, \dots, a_n\}$. If $\psi_I$ is an isomorphism for each `short interval' $I = (a_{i-1},a_{i+1})$ and $I = (a_{i}, a_{i+1})$ then $\psi$ is a natural isomorphism; that is, an isomorphism of cosheaves.
\end{cor}

\begin{proof}
We construct the inverse $\omega: \Gfunc \tO \Ffunc$ by setting
$\omega_I = \psi_I^{-1}$ for short intervals $I$ and applying the proposition (consistency follows from the naturality of~$\psi$). Since $(\omega \psi)_I = \mathbb{1}_{\Ffunc(I)}$ and $(\psi\omega)_I = \mathbb{1}_{\Gfunc(I)}$ for short intervals, the proposition implies $\omega\psi = \mathbb{1}_{\Ffunc}$ and $\psi\omega = \mathbb{1}_{\Gfunc}$.
Thus $\psi$ is a natural isomorphism.
\end{proof}

\subsection{The Reeb cosheaf functor $\CC$}
\label{sec:functor-C}

The {\bf Reeb cosheaf functor}, $\CC$, converts an $\R$-space to its {\bf Reeb cosheaf}. Let $f = (\X,f)$ be an $\R$-space. Then $\CC(f) = \Ffunc$ is the pre-cosheaf defined by
\[
\Ffunc(I) = \pi_0 f^{-1}(I),
\quad
\Ffunc[I \subseteq J] = \pi_0[ f^{-1}(I) \subseteq f^{-1}(J) ].
\]
(This is the second item in Example~\ref{ex:pre-cosheaves}.)

Any morphism $\alpha: (\X,f) \to (\Y,g)$ yields a natural transformation of pre-cosheaves $\CC[\alpha]: \Ffunc \tO \Gfunc$ defined as follows. For each interval~$I$, the map~$\alpha$ carries $f^{-1}(I)$ into $g^{-1}(I)$. We set
\begin{align*}
\CC[\alpha]_I = \pi_0[ f^{-1}(I) \stackrel{\alpha}{\to} g^{-1}(I)]
\quad\text{which is a map}\quad
\Ffunc(I) \to \Gfunc(I).
\end{align*}
We prove the naturality of $\CC[\alpha]$ by applying~$\pi_0$ to the commutative square on the left:
\[
\begin{diagram}
\node{f^{-1}(I)}
	\arrow{e,t}{\subseteq}
	\arrow[2]{s,l}{\alpha}
\node{f^{-1}(J)}
	\arrow[2]{s,r}{\alpha}
\node[2]{\Ffunc(I)}
	\arrow{e,t}{\Ffunc[I \subseteq J]}
	\arrow[2]{s,l}{\CC[\alpha]_I}
\node{\Ffunc(J)}
	\arrow[2]{s,r}{\CC[\alpha]_J}
\\
\node[3]{\stackrel{\pi_0}{\longrightarrow}}
\\
\node{g^{-1}(I)}
	\arrow{e,t}{\subseteq}
\node{g^{-1}(J)}
\node[2]{\Gfunc(I)}
	\arrow{e,t}{\Gfunc[I \subseteq J]}
\node{\Gfunc(J)}
\end{diagram}
\]
Since $\CC[-]$ respects composition and identities, we have a functor from $\R$-spaces to pre-cosheaves.

\begin{prop}
\label{prop:Cf-cosheaf}
The pre-cosheaf $\Ffunc = \CC(f)$ is a cosheaf.
\end{prop}

\begin{proof} 
Let $U$ be an open interval and let $(I_p \mid p \in A)$ be a cover of~$U$ by open intervals. We show that $\Ffunc(U)$ satisfies the universal property for the colimit of~\eqref{eq:cosheaf}. Accordingly, let $Z$ be a set and let $\zeta_p: \Ffunc(I_p) \to Z$ be functions satisfying the consistency condition~\eqref{eq:cosheaf-3}. We show that there is a unique $\zeta: \Ffunc(U) \to Z$ satisfying Eqn.~\eqref{eq:cosheaf-4}. 

Let $[y]_U$ denote the path-component of a point $y$ in $f^{-1}(U)$. Since $f(y)$ must belong to some~$I_p$, we are forced to define
\[
\zeta([y]_U)
= \zeta \circ \Ffunc[I_p \subseteq U] ([y]_{I_p}) 
= \zeta_p([y]_{I_p})
\]
and moreover the right-hand side does not depend on the choice of~$p$, because
\[
\zeta_p([y]_{I_p})
=
\zeta_p \circ \Ffunc[I_p \cap I_q \subseteq I_p]([y]_{I_p \cap I_q})
=
\zeta_q \circ \Ffunc[I_p \cap I_q \subseteq I_q]([y]_{I_p \cap I_q})
=
\zeta_q([y]_{I_q})
\]
by condition~\eqref{eq:cosheaf-3}, if $f(y) \in I_p \cap I_q$.

It remains to show that this definition is independent of the point~$y$ used to identify the component. Suppose $[y_0]_U = [y_1]_U$. Then there is  a continuous path $(y_t)$ in $f^{-1}(U)$ from $y_0$ to~$y_1$. Now every point in $[0,1]$ has a neighbourhood over which $f(y_t)$ is contained in some fixed~$I_p$. Over that neighbourhood, $[y_t]_{I_p}$ is constant and therefore $\zeta_p([y_t]_{I_p})$ is constant.
Since $[0,1]$ is connected, this local constancy implies global constancy and so $\zeta([y_0]_U) = \zeta([y_1]_U)$.
\end{proof}

\begin{prop}
\label{prop:Cf-constructible}
If $f$ is a constructible $\R$-space then $\Ffunc = \CC(f)$ is a constructible cosheaf with the same critical set.
\end{prop}

\begin{proof}
Let $I \subseteq J$ be intervals which meet~$S$ in the same set of points. The product structure over the components of $\R \setminus S$ implies that $f^{-1}(I) \subseteq f^{-1}(J)$ is a homotopy equivalence and therefore $\Ffunc[I \subseteq J]$ is an isomorphism.
\end{proof}

It follows from Propositions \ref{prop:Cf-cosheaf} and~\ref{prop:Cf-constructible} that the operation~$\CC$ defines functors as follows:
\[
\CC : \Rtop \to \Csh,
\quad
\CC' : \Rtopc \to \Cshc,
\quad
\CC'' : \Reeb \to \Cshc.
\]
We use the symbols $\CC'$ and $\CC''$ when we wish to be precise about the domain and range of our functors. When that is not important, we simply write~$\CC$.

\begin{thm}
\label{Thm:CommutativeTriangle}
The functors $\CC'$ and $\CC''\RR: \Rtopc \to \Cshc$ are naturally isomorphic.
\end{thm}

In other words when starting with a constructible $\R$-space, we can immediately use~$\CC$ to convert it to a cosheaf or we can take its geometric Reeb graph and then use~$\CC$ to convert it to a cosheaf; either way the result is the same.

\begin{proof}
First, we compare $\CC, \CC\RR$ regarded as functors $\Rtop \to \Csh$. There is a natural transformation $\eta = \CC\rho$ defined by applying the functor~$\CC$ to the canonical projection~$\rho$ (Observation~\ref{obs:Reeb-rho}). Specfically:
\[
\eta_f
= \CC [\rho_f]
= \CC[ (\X,f) \to (\X_f, \bar{f}) ]
\;:\;
\CC(f) \to \CC\RR(f) = \CC(\bar{f})
\]

We must show that $\eta$ is an isomorphism at each constructible $\R$-space, meaning that it restricts to a natural isomorphism $\CC' \tO \CC''\RR$.
Let $f = (\X,f)$ be constructible; then the cosheaves $\CC(f)$, $\CC\RR(f)$ are themselves constructible with the same critical set.
To show that $\eta_f: \CC(f) \tO \CC\RR(f)$ is an isomorphism it is enough, by Corollary~\ref{cor:cshc-isomorphism}, to show that $(\eta_f)_I: [\CC(f)](I) \to [\CC\RR(f)](I)$ is an isomorphism whenever $I$ is a short interval.

The cases $I = (a_{i-1},a_{i+1})$ are handled by the left diagram, and the cases $I = (a_{i}, a_{i+1})$ are handled by the right diagram:
\[
\begin{diagram}
\dgARROWLENGTH=2em
\node{\V_i}
	\arrow{e,t}{\simeq}
	\arrow{s}
\node{f^{-1}(I)}
	\arrow{s}
\\
\node{V_i}
	\arrow{e,t}{\simeq}
\node{\bar{f}^{-1}(I)}
\end{diagram}
\qquad
\qquad
\begin{diagram}
\dgARROWLENGTH=2em
\node{\E_i}
	\arrow{e,t}{\simeq}
	\arrow{s}
\node{f^{-1}(I)}
	\arrow{s}
\\
\node{E_i}
	\arrow{e,t}{\simeq}
\node{\bar{f}^{-1}(I)}
\end{diagram}
\]
This is~\eqref{eq:transfer-2} from the cylinder principle (Lemma~\ref{lem:cylinder}) for the canonical projection $(\X,f) \to (\X_f,\bar{f})$.
The horizontal inclusions are homotopy equivalences, and the left-hand map of each diagram induces a bijection of path-components, so the same is true of the right-hand maps.
In each case, we see that $(\eta_f)_I = \pi_0[f^{-1}(I) \to \bar{f}^{-1}(I)]$ is an isomorphism.
\end{proof}

We round out this section with a result promised earlier and two cautionary examples.

\begin{prop}[combinatorial description of constructible cosheaves]
\label{prop:cshc-data}
Given a critical set $S = \{a_0, a_1, \dots, a_n\}$, finite sets
\[
V_0, \dots, V_n
\quad
\text{and}
\quad
E_0, \dots, E_{n-1}
\]
and maps
\[
\ell_i: E_i \to V_i
\quad
\text{and}
\quad
r_i: E_i \to V_{i+1}
\]
for all~$i$. Then there exists a constructible cosheaf~$\Ffunc$ with critical set~$S$, together with (in Notation~\ref{notn:cosheaf-comb}) bijections $V^\Ffunc_i \cong V_i$ and $E^\Ffunc_i \cong E_i$ such that the maps $\ell^\Ffunc_i$ and~$r^\Ffunc_i$ correspond to the maps $\ell_i$ and~$r_i$.
Any two such cosheaves are canonically isomorphic.
\end{prop}

This fact, together with Proposition~\ref{prop:cshc-morphism}, amount to a particular instance of a more general result of MacPherson that we discuss in Section~\ref{sec:discussion}.

\begin{proof}
We may construct $\Ffunc$ as the Reeb cosheaf of the $\R$-graph with critical set~$S$ constructed from the same combinatorial data. The conditions on~$\Ffunc$ are easily verified; and it is a cosheaf with critical set~$S$, by Propositions \ref{prop:Cf-cosheaf} and~\ref{prop:Cf-constructible}.
Given another such cosheaf~$\Gfunc$, we have canonical identifications $V^\Ffunc_i \cong V_i \cong V^\Gfunc_i$ and $E^\Ffunc_i \cong E_i \cong E^\Gfunc_i$ through which $\ell^\Ffunc_i, r^\Ffunc_i$ correspond to $\ell_i, r_i$ and then to $\ell^\Gfunc_i, r^\Gfunc_i$. These identifications define an isomorphism of cosheaves, by Corollary~\ref{cor:cshc-isomorphism}.
\end{proof}

\begin{rmk}
For readers more familiar with category theory, the existence of the cosheaf~$\Ffunc$ may be proved more directly (i.e.
 without manufacturing a topological space) as follows. The starting data specifies the value of the cosheaf on every open interval that meets at most one critical point. Every interval~$I$ is the union of its subintervals of this type. We define $\Ffunc(I)$ to be the colimit associated to this union, and then invoke general properties of colimits to show that $\Ffunc$ is a cosheaf with the desired values on short intervals.
\end{rmk}

The good behavior of our functors on constructible $\R$-spaces does not extend to the non-constructible case. Here are two counterexamples based on the \emph{topologist's sine curve}
\[
\mathbb{S} = \{(x, \sin (1/x)) \mid 0 < x \leq 1\} \cup \{ (0, y) \mid y \in [-1,1]\}
\]
which is a connected but not path-connected compact subset of the plane.

\begin{ex}
\label{ex:non-cosheaf}
Proposition~\ref{prop:Cf-cosheaf} fails if we replace the path-component functor~$\pi_0$ by the connected-component functor~$\overline{\pi}_0$. 
Consider $(\mathbb{S},y)$ where $y$ is the projection onto the second coordinate.
Now $\mathbb{S}$ itself is connected. On the other hand, if $I$ is an interval that meets but does not contain $[-1,1]$ then $f^{-1}(I)$ consists of countably many connected components, one of which is the segment on the $y$-axis. Now cover the real line by intervals of that type. The associated colimit has at least two elements, since the segment on the $y$-axis is always separate from everything else. This breaks the cosheaf condition, since $\overline\pi_0 f^{-1}(\R) = \overline\pi_0(\mathbb{S})$ is a singleton.
%
%
\end{ex}

\begin{ex}
\label{ex:bad-CR}
The natural isomorphism of Theorem~\ref{Thm:CommutativeTriangle} does not extend to arbitrary $\R$-spaces. 
Consider $(\mathbb{S},x)$ where $x$ is the projection onto the first coordinate. We can identify its Reeb graph as follows: each levelset over $[0,1]$ is path-connected, so the projection to the $x$-axis induces a continuous bijection, and therefore homeomorphism, from the Reeb graph $\mathbb{S}_x$ to the interval $[0,1]$.
Then $\pi_0(\mathbb{S}_x) = \pi_0([0,1])$ is a singleton, whereas $\pi_0(\mathbb{S})$ is not. We deduce that the cosheaves $\CC(\mathbb{S},x)$ and $\CC\RR(\mathbb{S},x) = \CC(\mathbb{S}_x, x)$ are not isomorphic: they return non-isomorphic sets when evaluated at the interval~$\R$ (or indeed any interval containing~0).
\end{ex}

\subsection{Equivalence of categories}
\label{sec:equivalence}

This is the theorem that relates Reeb graphs to Reeb cosheaves.

\begin{thm}
\label{Thm:Equivalence}
The functor $\CC'': \Reeb \to \Cshc$ is an equivalence of categories.
\end{thm}

\begin{proof}
We will show that $\CC'' = \CC|_{\Reeb}$ is \emph{fully faithful} and \emph{essentially surjective}. This means:
\begin{itemize}
\item[(i)]
For every $f = (\X,f)$ and $g = (\Y,g)$  in $\Reeb$, the map
\[
\begin{diagram}
\node{\Hom(f,g)}
	\arrow{e,t}{\CC[-]}
\node{\Hom(\CC(f), \CC(g))}
\end{diagram}
\]
is a bijection of sets.

\item[(ii)]
For every $\Ffunc \in \Cshc$ there exists $f = (\X,f)$ in $\Reeb$ such that $\CC(f)$ is isomorphic to~$\Ffunc$.
\end{itemize}
It is a theorem~\cite[$\S$IV.4]{MacLane_1998} that such a functor is an equivalence of categories.

\medskip
(i)
We show that $\CC''$ is fully faithful.
Let $f \in \Reeb$ and write $\Ffunc = \CC(f)$. With respect to a critical set for both, we can describe $f$ and $\Ffunc$ by data
\begin{alignat*}{2}
V_i^f, E_i^f, \ell_i^f, r_i^f
&\qquad \text{(Notation~\ref{notn:Reeb-comb})}
\\
V_i^{\Ffunc}, E_i^{\Ffunc}, \ell_i^{\Ffunc}, r_i^{\Ffunc}
&\qquad \text{(Notation~\ref{notn:cosheaf-comb})}
\end{alignat*}
respectively. Applying $\pi_0$ to the homotopy equivalences from the cylinder principle (Lemma~\ref{lem:cylinder}) and to diagram~\eqref{eq:transfer-1}, we get isomorphisms
\begin{alignat*}{3}
V_i^f
&= \pi_0(V_i^f)
&&\,= \pi_0 f^{-1}(a_{i-1},a_{i+1})
&&\,= V^\Ffunc_i
\\
E_i^f
&= \pi_0(E_i^f)
&&\,= \pi_0 f^{-1}(a_{i},a_{i+1})
&&\,= E^\Ffunc_i
\end{alignat*}
which carry $\ell_i^f, r_i^f$ to $\ell_i^\Ffunc, r_i^\Ffunc$.

Now let $f, g \in \Reeb$ and write $\Ffunc = \CC(f)$, $\Gfunc = \CC(g)$. With respect to a common critical set, we can describe $f,g$ and $\Ffunc, \Gfunc$ by corresponding sets of data, with isomorphisms as above. From the characterizations of morphisms given in Propositions \ref{Prop:MapData} and~\ref{prop:cshc-morphism}, there is an obvious bijection between $\Hom(f,g)$ and $\Hom(\CC(f),\CC(g))$
\[
\bigg\{
\begin{array}{rl}
\phi_i^V: & V_i^f \to V_i^g
\\
\phi_i^E: & E_i^f \to E_i^g
\end{array}
\bigg\}
\quad
\longleftrightarrow
\quad
\bigg\{
\begin{array}{rl}
\psi_i^V: & V_i^{\Ffunc} \to V_i^{\Gfunc}
\\
\psi_i^E: & E_i^{\Ffunc} \to E_i^{\Gfunc}
\end{array}
\bigg\}
\]
defined using these isomorphisms.
To show that this bijection is given by $\CC[-]$ we apply $\pi_0$ to the diagrams~\eqref{eq:transfer-2} from the cylinder principle.
This completes the proof that $\CC''$ is fully faithful.

\medskip
(ii)
We show that $\CC''$ is essentially surjective.
Let $\Ffunc \in \Cshc$ with critical set $S = \{a_0, a_1, \dots, a_n\}$.
Let $(\X,f)$ be the $\R$-graph with critical set~$S$ defined by the following data (Notation~\ref{notn:Reeb-comb})
\begin{alignat}{2}
\label{eq:DD-obj-1}
V_i
&= \Ffunc((a_{i-1},a_{i+1}))
\qquad
&\ell_i
&= \Ffunc[(a_i, a_{i+1}) \subseteq (a_{i-1},a_{i+1})]
\\
\label{eq:DD-obj-2}
E_i
&= \Ffunc((a_{i},a_{i+1}))
&r_i
&= \Ffunc[(a_i, a_{i+1}) \subseteq (a_{i},a_{i+2})]
\end{alignat}
and write $\Ffunc' = \CC(f)$.

From the cylinder principle (Lemma~\ref{lem:cylinder}) and diagram~\eqref{eq:transfer-1} we obtain isomorphisms
\begin{alignat}{3}
\label{eq:DD-mor-1}
&\Ffunc'((a_{i-1},a_{i+1}))
&&\,= V_i
&&\,= \Ffunc((a_{i-1},a_{i+1}))
\\
\label{eq:DD-mor-2}
&\Ffunc'((a_{i},a_{i+1}))
&&\,= E_i
&&\,= \Ffunc((a_{i},a_{i+1}))
\end{alignat}
which are natural with respect to the inclusions $(a_i,a_{i+1}) \subseteq (a_{i-1}, a_{i+1})$ and $(a_i,a_{i+1}) \subseteq (a_{i}, a_{i+2})$.
Then Corollary~\ref{cor:cshc-isomorphism} implies that the cosheaves $\Ffunc$ and $\Ffunc' = \CC(f)$ are isomorphic.
\end{proof}

Since $\CC''$ is an equivalence of categories, it has an inverse functor $\DD: \Cshc \to \Reeb$ called the {\bf display locale functor}~\cite{Funk1995}.
There are various ways to define this functor; the result is unique up to a canonical natural isomorphism. We can define $\DD$ combinatorially as follows: given $\Ffunc \in \Cshc$ with critical set~$S$, let $\DD(\Ffunc)$ be the $\R$-graph defined by the data \eqref{eq:DD-obj-1} and~\eqref{eq:DD-obj-2}. A natural isomorphism $\CC\DD \tO \onefunc_{\Cshc}$ is defined by the identifications in \eqref{eq:DD-mor-1} and~\eqref{eq:DD-mor-2}.
These identifications uniquely determine the result $\DD[\alpha]$ of applying $\DD$ to a morphism~$\alpha$.

\begin{cor}
\label{cor:CRD}
The functors $\DD\CC', \RR: \Rtopc \to \Reeb$ are naturally isomorphic.
\end{cor}

That is, the Reeb graph of a constructible $\R$-space is equal to the display locale of its Reeb cosheaf.

\begin{proof}
We have $\RR \simeq \DD \CC'' \RR \simeq \DD \CC'$ by Theorems \ref{Thm:Equivalence} and~\ref{Thm:CommutativeTriangle}.
\end{proof}

\begin{rmk}
The display locale of a cosheaf can be defined more abstractly and generally~\cite{Funk1995}, yielding a functor $\DD: \Csh \to \Rtop$.
The fiber of $\DD(\Ffunc)$ at $a \in \R$ is defined to be the limit of $\Ffunc(I)$ over intervals $I$ containing~$a$. This is called the \emph{co-stalk} of the cosheaf at~$a$. The disjoint union of these co-stalks is topologized as follows: for each interval~$I$ and $x \in \Ffunc(I)$ there is a basic open set $U_{I,x}$ defined to be the elements of the co-stalks at all~$ a \in I$ which project to~$x$.
\end{rmk}

\section{The interleaving distance}
\label{Sect:CshInterleaving}

We are ready to define the distance between a pair of Reeb graphs. The abstract principle is quite simple (Section~\ref{Sect:funct-interleaving}): we regard the Reeb graphs as constructible cosheaves,  then we compare the cosheaves using an `interleaving distance'~\cite{Chazal2009b,Bubenik2014}. Of course, we would like to interpret this as geometrically as possible.
To do this, we consider two parallel operations: {smoothing} of pre-cosheaves (Section~\ref{Sect:smoothing}) and {thickening} of $\R$-spaces (Section~\ref{Sect:thickening}). The interleaving distance may be expressed in terms of smoothings; the resulting distance on Reeb graphs may be expressed in terms of thickenings.

The smoothing functors $(\SS_\e)$ and the thickening functors $(\TT_\e)$ give compatible transformations on the two sides of our road map: see Figure~\ref{Fig:MasterDiagram2}.
\begin{figure}
\begin{equation*}
\xymatrix{
 & \Pre
 	\ar@(dr,r)_{\SS_\e}\\
 \Rtop
 	\ar[r]^{\CC}
	\ar@(l,ul)^{\TT_\e}
 & \Csh
 	\ar[u]
	\ar@(dr,r)_{\SS_\e}
\\
 \Rtopc
 	\ar[r]^{\CC'}
	\ar[u]
	\ar@<2pt>[d]^{\RR}
	\ar@(l,ul)^{\TT_\e}
& \Cshc
	\ar[u]
	\ar@(dr,r)_{\SS_\e}
\\
 \Reeb
	\ar@<2pt>[r]^{\CC''}
 	\ar@<2pt>[u]
	\ar@(l,ul)^{\UU_\e}
 & \Cshc
 	\ar@{=}[u]
	\ar@<2pt>[l]^\DD 
}
\end{equation*}
\caption[Road map with smoothing and thickening functors]{Road map with functors for smoothing cosheaves $(\SS_\e)$, for thickening $\R$-spaces $(\TT_e)$, and for smoothing $\R$-graphs $(\UU_\e)$.
}
\label{Fig:MasterDiagram2}
\end{figure}
Smoothing preserves the subcategories of cosheaves and constructible cosheaves. In a sense, that explains why it has geometric significance and why the existence of the thickening functor is not a surprise.
Thickening preserves the subcategory of constructible $\R$-spaces.
This allows us to define a semigroup of topological smoothing functors $(\UU_\e)$ on Reeb graphs (Section~\ref{Sect:smooth-Reeb}).

\subsection{Interleaving of pre-cosheaves}
\label{Sect:funct-interleaving}

Interleavings of persistence modules were used, implicitly, in the proof of the persistence stability theorem of Cohen-Steiner et al.~\cite{Cohen-Steiner2007}. Chazal et al.\ defined the concept explicitly for their algebraic stability theorem~\cite{Chazal2009b}. More recently Bubenik and Scott have given a general formulation in categorical language~\cite{Bubenik2014}; we follow their ideas closely.

Interleavings are approximate isomorphisms. Let $\Ffunc, \Gfunc: \Int \to \Set$ be pre-cosheaves. Recall that an isomorphism between them consists of two families of maps
\[
\phi_I: \Ffunc(I) \to \Gfunc(I),
\quad
\psi_I: \Gfunc(I) \to \Ffunc(I)
\]
that are natural with respect to inclusions $I \subseteq J$, such that $\phi_I, \psi_I$ are inverses for all~$I$.

We can give ourselves $\e$~leeway by expanding the intervals slightly. For any interval $I = (a,b)$, let $I^\e = (a-\e,b+\e)$ denote the interval expanded by~$\e \geq 0$.

\begin{defn}
\label{def:interleaving}
An {\bf $\e$-interleaving} between $\Ffunc, \Gfunc$ is given by two families of maps
\begin{equation}
\label{eq:interleaving1}
\phi_I: \Ffunc(I) \to \Gfunc(I^\e),
\quad
\psi_I: \Gfunc(I) \to \Ffunc(I^\e)
\end{equation}
which are natural with respect to inclusions $I \subseteq J$ and such that
\begin{equation}
\label{eq:interleaving2}
\psi_{I^\e} \circ \phi_I = \Ffunc[I \subseteq I^{2\e}],
\quad
\phi_{I^\e} \circ \psi_I = \Gfunc[I \subseteq I^{2\e}]
\end{equation}
for all $I$. When $\e=0$ this is is precisely an isomorphism between $\Ffunc, \Gfunc$.
\end{defn}

\begin{defn}
\label{def:interleaving-dist}
The {\bf interleaving distance} between two co-presheaves $\Ffunc, \Gfunc: \Int \to \Set$ is defined
\[
\inter(\Ffunc,\Gfunc) =
\inf\left( \e \mid \text{there exists an $\e$-interleaving between $\Ffunc,\Gfunc$} \right).
\]
The {\bf Reeb distance} between two $\R$-graphs $f = (\X,f)$ and $g = (\Y,g)$ is defined
\[
\rdist(f,g) = \inter(\CC(f), \CC(g)).
\]
(The infimum of an empty set is understood to be~$\infty$.)
\end{defn}

This definition of~$\rdist$ may not seem immediately helpful: it requires converting the Reeb graphs into cosheaves, and then comparing the cosheaves by a metric that is itself somewhat mysterious.
In the next few sections we will develop a more geometric formulation that allows us---at least in principle---to compute the distance function.
Having said that, certain geometric assertions are immediately accessible. We present some of these results now.

\begin{prop}
\label{prop:inter}
The interleaving distance $\inter$ is an extended pseudometric on $\Pre$: it takes values in $[0,\infty]$, it is symmetric, it satisfies the triangle inequality, and $\inter(\Ffunc,\Ffunc) = 0$.
It follows that $\rdist$ is an extended pseudometric on $\Reeb$.
\end{prop}

\begin{proof}
For the triangle inequality, note that if $(\phi^1_I)$, $(\psi^1_I)$ define an $\e_1$-interleaving between $\Ffunc, \Gfunc$ and $(\phi^2_I)$, $(\psi^2_I)$ define an $\e_2$-interleaving between $\Gfunc, \Hfunc$ then
\[
\phi^3_{I} = \phi^2_{I^{\e_1}} \circ \phi^1_I,
\quad
\psi^3_{I} = \phi^1_{I^{\e_2}} \circ \psi^2_I
\]
define an $(\e_1+\e_2)$-interleaving between $\Ffunc, \Hfunc$. The remaining statements are obvious.
\end{proof}

This approach to interleaving distances~\cite{Chazal2009b,Bubenik2014} is designed to make the next theorem as easy as possible.

\begin{thm}[Stability of Reeb distance]
\label{thm:stability}
(i) Let $(\X,f)$ and $(\X,g)$ be $\R$-spaces (with the same total space~$\X$). Then:
\[
\inter(\CC(f),\CC(g)) \leq \| f - g \|_\infty
\]
(ii)
Let $(\X,f)$ and $(\X,g)$ be constructible $\R$-spaces. Then the Reeb graphs $(\X_f, \bar{f})$ and $(\X_g, \bar{g})$ satisfy:
\[
\rdist(\bar{f}, \bar{g}) \leq \| f - g \|_\infty
\]
\end{thm}

\begin{proof}
(i) Suppose $\| f - g \|_\infty \leq \e$. We show that there is an $\e$-interleaving between $\CC(f), \CC(g)$. The supremum bound implies that we have inclusions
\[
f^{-1}(I) \subseteq g^{-1}(I^\e),
\quad
g^{-1}(I) \subseteq f^{-1}(I^\e)
\]
for all~$I$. Accordingly, we define
\[
\phi_I = \pi_0[f^{-1}(I) \subseteq g^{-1}(I^\e)],
\quad
\psi_I = \pi_0[g^{-1}(I) \subseteq f^{-1}(I^\e)].
\]
Naturality and the other conditions are satisfied because diagrams of inclusions always commute.

\medskip
(ii) This follows from part~(i) because the natural isomorphism $\CC''\RR = \CC'$ (Theorem~\ref{Thm:CommutativeTriangle}) implies that $\CC(\bar{f}), \CC(\bar{g})$ are isomorphic to $\CC(f), \CC(g)$.
\end{proof}

The Reeb distance is sometimes infinite.

\begin{prop}
\label{Prop:IsoReeb}
The Reeb distance between two $\R$-graphs $(\X, f)$, $(\Y, g)$ is finite if and only if $\X,\Y$ have the same number of path components.
\end{prop}

\begin{proof}
Let $I$ be a very large interval containing $f(\X) \cup g(\Y)$. Writing $\Ffunc = \CC(f)$ and $\Gfunc = \CC(g)$ we have
\[
\pi_0(\X) = \Ffunc(I) = \Ffunc(I^\e) = \Ffunc(I^{2\e}),
\quad
\pi_0(\Y) = \Gfunc(I) = \Gfunc(I^\e) = \Gfunc(I^{2\e})
\]
for any $\e \geq 0$. Thus any $\e$-interleaving defines a bijection $\pi_0(\X) \cong \pi_0(Y)$ through its maps $\phi_I, \psi_I$.

Conversely, suppose there is a bijection $\pi_0(\X) \cong \pi_0(\Y)$. Let $\e$ be larger than the diameter of $f(\X) \cup g(\Y)$. This implies that if $I$ meets either $f(\X)$ or $g(\Y)$ then $\Ffunc(I^\e) = \pi_0(\X)$ and $\Gfunc(I^\e) = \pi_0(\Y)$.
For these intervals we define $\phi_I, \psi_I$ as the following composites, using the bijection.
\begin{alignat*}{5}
\phi_I &:
\Ffunc(I) 
  \stackrel{\Ffunc[I \subseteq I^\e]}{\xrightarrow{\hspace*{2.5em}}}
\Ffunc(I^\e)
&&= \pi_0(\X)
&&\cong \pi_0(\Y)
&&= \Gfunc(I^\e)
\\
\psi_I &:
\Gfunc(I)
  \stackrel{\Gfunc[I \subseteq I^\e]}{\xrightarrow{\hspace*{2.5em}}}
\Gfunc(I^\e)
&&= \pi_0(\Y)
&&\cong \pi_0(\X)
&&= \Ffunc(I^\e)
\end{alignat*}
For the remaining intervals, $\Ffunc(I) = \Gfunc(I) = \emptyset$ so nothing needs to be done.
It is not difficult to verify that these maps define an $\e$-interleaving.
\end{proof}

\begin{prop}
\label{prop:rdist=0}
The Reeb distance between two $\R$-graphs is zero if and only if they are isomorphic.
\end{prop}

\begin{proof}
The nontrivial part is to show that Reeb distance zero implies that the $\R$-graphs are isomorphic.
We will show that if $\Ffunc, \Gfunc$ are constructible cosheaves with $\inter(\Ffunc, \Gfunc) = 0$ then they are isomorphic. This is an equivalent statement since $\CC''$ is an equivalence of categories.

Here is a quantified statement that implies the result.
Let $S = \{a_0, a_1, \dots, a_n\}$ be a common critical set for $\Ffunc, \Gfunc$ and let $\hbar = \min_i(a_{i+1}-a_{i}) > 0$. We claim that if $\Ffunc, \Gfunc$ are $\e$-interleaved for $\e < \hbar/4$ then $\Ffunc, \Gfunc$ are isomorphic.

Indeed, for such~$\e$ we can find intervals $I_0, I_1, \dots, I_n$ such that
\[
a_i \in I_i \subseteq I_i^{2\e} \subseteq (a_{i-1}, a_{i+1})
\]
where each $J_i = I_i \cap I_{i+1}$ is nonempty.\footnote{%
Specifically, $I_i =  \left( \tfrac{1}{2}(a_{i-1} + a_i) - \delta, \tfrac{1}{2}(a_{i} + a_{i+1}) + \delta\right)$ will do, for sufficiently small~$\delta > 0$.
}
Note that $J_i \subseteq J_i^{2\e} \subseteq (a_i,a_{i+1})$ automatically.
By the constructibility of $\Ffunc, \Gfunc$ the various inclusions of intervals induce isomorphisms
\begin{alignat*}{7}
\Ffunc(I_i) &= 
\Ffunc(I_i^\e) &&= 
\Ffunc(I_i^{2\e}) &&=
\Ffunc((a_{i-1},a_{i+1})),
\quad
&&
\Ffunc(J_i) &&= 
\Ffunc(J_i^\e) &&= 
\Ffunc(J_i^{2\e}) &&=
\Ffunc((a_{i},a_{i+1})),
\\
\Gfunc(I_i) &= 
\Gfunc(I_i^\e) &&= 
\Gfunc(I_i^{2\e}) &&=
\Gfunc((a_{i-1},a_{i+1})),
\quad
&&
\Gfunc(J_i) &&= 
\Gfunc(J_i^\e) &&= 
\Gfunc(J_i^{2\e}) &&=
\Gfunc((a_{i},a_{i+1})).
\end{alignat*}
It follows that the following maps from an $\e$-interleaving
\[
\begin{diagram}
\dgARROWLENGTH=1.5em
\node{\Ffunc(I_i)}
	\arrow{e,=}
	\arrow{se}
\node{\Ffunc(I_i^\e)}
	\arrow{e,=}
	\arrow{se}
\node{\Ffunc(I_i^{2\e})}
\\
\node{\Gfunc(I_i)}
	\arrow{e,=}
	\arrow{ne}
\node{\Gfunc(I_i^\e)}
	\arrow{e,=}
	\arrow{ne}
\node{\Gfunc(I_i^{2\e})}
\end{diagram}
\qquad
\text{and}
\qquad
\begin{diagram}
\dgARROWLENGTH=1.5em
\node{\Ffunc(J_i)}
	\arrow{e,=}
	\arrow{se}
\node{\Ffunc(J_i^\e)}
	\arrow{e,=}
	\arrow{se}
\node{\Ffunc(J_i^{2\e})}
\\
\node{\Gfunc(J_i)}
	\arrow{e,=}
	\arrow{ne}
\node{\Gfunc(J_i^\e)}
	\arrow{e,=}
	\arrow{ne}
\node{\Gfunc(J_i^{2\e})}
\end{diagram}
\]
induce isomorphisms $\Ffunc((a_{i-1},a_{i+1})) \cong \Gfunc((a_{i-1},a_{i+1}))$ and $\Ffunc((a_{i},a_{i+1})) \cong \Gfunc((a_{i},a_{i+1}))$. 
These isomorphisms are natural with respect to the inclusions $(a_{i-1},a_i) \subseteq (a_{i-1},a_{i+1})$ and $(a_{i},a_{i+1}) \subseteq (a_{i-1},a_{i+1})$ because the interleaving maps are natural with respect to $J_{i-1} \subseteq I_i$ and $J_i \subseteq I_i$.

Proposition~\ref{prop:cshc-morphism} and Corollary~\ref{cor:cshc-isomorphism} imply that this collection of isomorphisms extends to an cosheaf isomorphism between $\Ffunc, \Gfunc$.
\end{proof}

\begin{cor}
The Reeb distance~$\rdist$ is an extended \emph{metric} on isomorphism classes in $\Reeb$.
\qed
\end{cor}

\subsection{Smoothing functors}
\label{Sect:smoothing}

We can express the notion of $\e$-interleaving in categorical language, following~\cite{Bubenik2014}. 
Think of the expansion operation on intervals
\[
\Omega_\e: \Int \to \Int;\; I \mapsto I^\e
\]
as a functor (since $I \subseteq J$ implies $I^\e \subseteq J^\e$).
Any pre-cosheaf $\Ffunc: \Int \to \Set$ can be `$\e$-smoothed' to obtain a new pre-cosheaf $\Ffunc \Omega_\e: \Int \to \Set$. Thus $\Ffunc\Omega_\e(I) = \Ffunc(I^\e)$ by definition.

\begin{obs}
\label{obs:smooth-1}
Asking for natural families of maps $(\phi_I), (\psi_I)$ as in equation~\eqref{eq:interleaving1} is precisely the same as asking for natural transformations $\phi: \Ffunc \tO \Gfunc\Omega_\e$ and $\psi: \Gfunc \tO \Ffunc\Omega_\e$.
\end{obs}

\begin{obs}
\label{obs:smooth-2}
The functor~$\Omega_\e$ is a {pointed endofunctor} on~$\Int$, because the inclusions $I \subseteq I^\e$ define a natural transformation $\omega^\e: \onefunc_{\Int} \tO \Omega_\e$.
From this we get a natural transformation
\[
\sigma^\e_\Ffunc = \Ffunc \omega^\e: \Ffunc \tO \Ffunc\Omega_\e
\]
defined explicitly by $(\sigma^\e_\Ffunc)_I = \Ffunc[I \subseteq I^\e] : \Ffunc(I) \to \Ffunc(I^\e)$.
\end{obs}

\begin{obs}
\label{obs:smooth-3}
Conditions~\eqref{eq:interleaving2} can be written as
$(\psi \Omega_\e) \circ \phi = \sigma^{2\e}_{\Ffunc}$ and
$(\phi \Omega_\e) \circ \psi = \sigma^{2\e}_{\Gfunc}$.
\end{obs}

The two observations combined give a more purely categorical definition of $\e$-interleaving. The restatement of~\eqref{eq:interleaving2} asks that the following diagrams (of natural transformations) commute:
\begin{equation}
\label{eq:interleaving-cat}
\renewcommand{\labelstyle}{\textstyle}
\raisebox{9ex}{
\xymatrix
@=6ex
{
\Ffunc
	\ar@{=>}[dr]^(0.45){\phi}
	\ar@{=>}[dd]_(0.45){\sigma^{2\e}_{\Ffunc}}
\\
&
\Gfunc\Omega_\e
	\ar@{=>}[dl]^(0.45){\psi \Omega_\e} 
\\
\Ffunc\Omega_{2\e}
}}
\qquad
{\text{and}}
\qquad
\raisebox{9ex}{
\xymatrix
@=6ex
{
&
\Gfunc
	\ar@{=>}[dl]_(0.45){\psi}
	\ar@{=>}[dd]^(0.45){\sigma^{2\e}_{\Gfunc}}
\\
\Ffunc\Omega_\e
	\ar@{=>}[dr]_(0.45){\phi \Omega_\e} 
\\
&
\Gfunc\Omega_{2\e}
}}
\end{equation}
Implicitly we are using $\Omega_\e \Omega_\e = \Omega_{2\e}$.

\medskip
Now we consider the smoothing operation $\Ffunc \mapsto \Ffunc\Omega_\e$ in its own right.
This, we claim, is a functor on pre-cosheaves. Indeed, any natural transformation
$
\phi: \Ffunc \tO \Gfunc
$
gives rise to a natural transformation
$
\phi \Omega_\e: \Ffunc\Omega_\e \tO \Gfunc\Omega_\e
$
defined explicitly by $(\phi \Omega_\e)_I = \phi_{I^\e}: \Ffunc(I^\e) \to \Gfunc(I^\e)$.
One verifies immediately that this procedure respects composition and identities. Thus:

\begin{defn}[Smoothing functor]
Let $\SS_\e$ be the endofunctor of $\Pre = \Set^\Int$ defined by precomposition with $\Omega_\e$. Thus $\SS_\e(\Ffunc) = \Ffunc\Omega_\e$, and $\SS_\e[\phi] = \phi\Omega_\e: \Ffunc\Omega_\e\tO\Gfunc\Omega_\e$ for a morphism $\phi:\Ffunc \tO \Gfunc$.
\end{defn}

\begin{obs}
\label{obs:smooth-2b}
It follows from Observation~\ref{obs:smooth-2} that each $\SS_\e$ is a pointed endofunctor of $\Pre$, in the sense that there is a natural map~$\sigma^\e_\Ffunc$ from each pre-cosheaf~$\Ffunc$ to its smoothing $\SS_\e(\Ffunc)$. This is defined at each interval~$I$ to be
\[
\Ffunc[I \subseteq I^\e] : \Ffunc(I) \to \Ffunc(I^\e) = \SS_\e(\Ffunc)(I).
\]
Succinctly, these maps comprise a natural transformation $\sigma^\e: \onefunc_\Pre \tO \SS_\e$, defined $\sigma^\e_\Ffunc = \Ffunc \omega^\e$.
\end{obs}

In the remainder of this section, we study the properties of~$\SS_\e$. 
The next two propositions support an analogy between smoothing of pre-cosheaves and smoothing in functional analysis (for example by convolution with a heat kernel).

\begin{prop}
\label{prop:smoothing-semigroup}
$(\SS_\e)$ is a semigroup of endofunctors (on $\Pre$) because $(\Omega_\e)$ is a semigroup of endofunctors (on~$\Int$): the relation $\Omega_{\e_1+\e_2} = \Omega_{\e_2} \Omega_{\e_1}$ implies $\SS_{\e_1+\e_2} = \SS_{\e_1} \SS_{\e_2}$.
\qed
\end{prop}

\begin{prop}
\label{prop:smoothing-contraction}
Smoothing is a contraction: $\inter(\SS_\e(\Ffunc), \SS_\e(\Gfunc)) \leq \inter(\Ffunc, \Gfunc)$.
\end{prop}

\begin{proof}
If $\phi, \psi$ define a $\delta$-interleaving between $\Ffunc, \Gfunc$ then $\phi\Omega_\e, \psi\Omega_\e$ define an $\delta$-interleaving between $\SS_\e(\Ffunc) = \Ffunc\Omega_\e$ and $\SS_\e(\Gfunc) = \Gfunc\Omega_\e$.
\end{proof}

The  next theorem indicates that we can make geometric use of~$\SS_\e$. 

\begin{thm}
\label{thm:S-restrictions}
The functor $\SS_\e$ restricts to functors $\SS_\e: \Csh \to \Csh$ and $\SS_\e: \Cshc \to \Cshc$.
\end{thm}

We split the theorem into two propositions.

\begin{prop}
\label{prop:smooth-constructible}
The functor $\SS_\e$ carries constructible pre-cosheaves to constructible pre-cosheaves.
\end{prop}

\begin{proof}
Let $S$ be a critical set for~$\Ffunc$. We claim that $S^\e:= (S-\e) \cup (S+\e)$ is a critical set for $\Ffunc\Omega_\e$.
Indeed, if $I \subseteq J$ and $J \setminus I$ does not meet $S^\e$ then $J^\e \setminus I^\e$ does not meet~$S$. Thus $\Ffunc\Omega_\e[I \subseteq J] = \Ffunc[I^\e \subseteq J^\e]$ is an isomorphism. And if $I$ is contained in $(-\infty, \min(S^\e)) \cup (\max(S^\e), +\infty)$ then $I^\e$ is contained in $(-\infty, \min(S)) \cup (\max(S),+\infty)$ so $\Ffunc\Omega_\e(I) = \Ffunc(I^\e) = \emptyset$.
\end{proof}

\begin{prop}
\label{prop:smooth-cosheaf}
The functor $\SS_\e$ carries cosheaves to cosheaves.
\end{prop}

\begin{proof}
Let $\Ffunc: \Int \to \Set$ be a cosheaf. We show that $\Ffunc\Omega_\e$ is also a cosheaf.

Let $U$ be an open interval covered by open intervals $(I_p \mid a \in A)$. Then $U^\e$ is covered by $(I_p^\e \mid a \in A)$.
We want to show that $\Ffunc\Omega^\e(U)$ is the colimit of the diagram
\[
\tag{$\dagger$}
\coprod_{p,q} \Ffunc \Omega_\e(I_p \cap I_q)
\rightrightarrows
\coprod_p \Ffunc \Omega_\e(I_p),
\]
knowing that $\Ffunc(U^\e)$ is the colimit of the diagram
\[
\tag{$\ddagger$}
\coprod_{p,q} \Ffunc (I^\e_p \cap I^\e_q)
\rightrightarrows
\coprod_p \Ffunc (I^\e_p).
\]
The two diagrams are almost identical, except that ($\ddagger$) has extra terms on the left-hand side whenever $I_p^\e \cap I_q^\e$ is nonempty but $I_p \cap I_q$ is empty. We will show that these extra terms do not affect the colimit.

Let $Z$ be a set, and suppose we are given maps $\zeta_p: \Ffunc\Omega_\e(I_p) = \Ffunc(I_p^\e) \to Z$ for all~$p$, such that
\[
\zeta_p \circ \Ffunc[I_p^\e \cap I_q^\e \subseteq I_p^\e]
=
\zeta_q \circ \Ffunc[I_p^\e \cap I_q^\e \subseteq I_q^\e]
\tag{$\S$}
\]
whenever $I_p \cap I_q \ne \emptyset$. We will show that equation~($\S$) holds in the additional cases where $I_p^\e \cap I_q^\e \ne \emptyset$. Then by the universal property for the colimit of~($\ddagger$) there will be a unique map $\zeta: \Ffunc(U^\e) = \Ffunc\Omega_\e(U) \to Z$ such that
\[
\zeta_p
= \zeta \circ \Ffunc[I_p^\e \subseteq U^\e]
= \zeta \circ \Ffunc\Omega_\e [I_p \subseteq U]
\]
and this will confirm that $\Ffunc\Omega_\e(U)$ satisfies the universal property for the colimit of~($\dagger$).

To this end, let $J = I_p^\e \cap I_q^\e$ be nonempty with $I_p \cap I_q = \emptyset$, so that the two open intervals sandwich between them a nonempty closed interval $K$. Since $K$ is compact and connected and contained in~$U$, we can find a finite connected chain $(I_{p_i})$ of intervals meeting~$K$ which connects the two ends of~$K$; so $I_p = I_{p_0}, I_{p_1}, \dots, I_{p_n} = I_q$ with each $I_{p_i} \cap I_{p_{i+1}}$ nonempty. Now, by metric considerations, each thickened interval $I_{p_i}^\e$ contains~$J$. Then for each $i$ we have the following diagram:
\[
\renewcommand{\labelstyle}{\textstyle}
\xymatrix@R=8ex@C=6ex{
&
\Ffunc(I_{p_i}^\e)
	\ar[dr]^{\zeta_{p_i}}
\\
\Ffunc(J)
	\ar[ur]
	\ar[r]
	\ar[dr]
&
\Ffunc(I_{p_i}^\e \cap I_{p_{i+1}}^\e)
	\ar[u]
	\ar[d]
&
Z
\\
&
\Ffunc(I_{p_{i+1}}^\e)
	\ar[ur]_{\zeta_{p_{i+1}}}
}
\]
The five maps on the left are those assigned by~$\Ffunc$ to the corresponding inclusions of intervals. The two triangles on the left commute since $\Ffunc$ is a functor, and the quadrilateral on the right commutes by~($\S$)
because $I_{p_i} \cap I_{p_{i+1}}$ is nonempty.
It follows that 
\[
\zeta_{p_i} \circ \Ffunc[J \subseteq I_{p_i}^\e]
=
\zeta_{p_{i+1}} \circ \Ffunc[J \subseteq I_{p_{i+1}}^\e]
\]
so, following the chain, we deduce ($\S$) for the pair $p,q$.

In sum, we have shown that if ($\S$) holds for all $p,q$ with $I_p \cap I_q$ nonempty, then it holds for all $p,q$ with $I_p^\e \cap I_q^\e$ nonempty. Thus the extra terms do not affect the colimit, and the proof is complete.
\end{proof}

\begin{rmk}
The proof is not specific to the category $\Set$. If $\Ffunc: \Int \to \mathbf{C}$ is a cosheaf in an arbitrary category~$\mathbf{C}$ (with an initial object), then its smoothing $\Ffunc \Omega_\e$ is a cosheaf, by the same argument.
\end{rmk}

\subsection{Thickening functors}
\label{Sect:thickening}

On the geometric side there is a family of functors $(\TT_\e)$ acting in parallel to the smoothing functors $(\SS_\e)$ that act on the cosheaf side. We study these functors now.

\begin{defn}
For $\e \geq 0$, the thickening functor $\TT_\e: \Rtop \to \Rtop$ is defined as follows.
\begin{itemize}
\item
Let  $(\X,f)$ be an $\R$-space. Then $\TT_\e(\X,f) = (\X_\e,f_\e)$ where $\X_\e = \X \times [-\e,\e]$ and $f_\e(x,t) = f(x) + t$.

\item
Let $\alpha: (\X,f) \to (\Y,g)$ be a morphism. Then $\TT_\e[\alpha]: (\X_\e,f_\e) \to (\Y_\e,f_\e); \,(x,t) \mapsto (\alpha(x),t)$.
\end{itemize}
It is easily confirmed that this is a functor.
\end{defn}

\begin{obs}
\label{obs:thick}
The thickening functor $\TT_\e$ is a pointed endofunctor of $\Rtop$. Indeed, the canonical embedding of $\X$ as the zero section of~$\X_\e$ defines a natural transformation  $\tau^\e: \onefunc_{\Rtop} \tO \TT_\e$.
Schematically we draw the picture
\[
\includegraphics[scale=0.5]{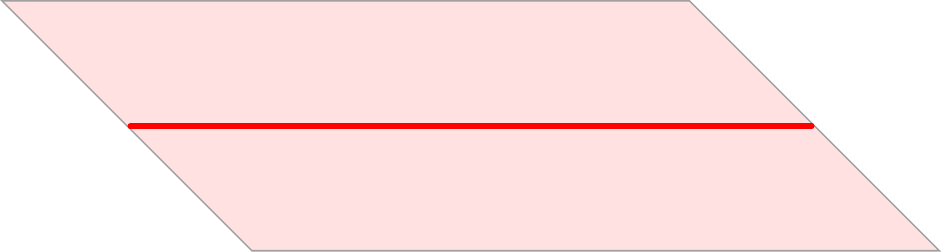}
\]
and formally we define $\tau^\e = (\tau^\e_f)$ by
\[
\tau^\e_f:
(\X,f) \to (\X_\e, f_\e);
\;
x \mapsto (x,0).
\]
Naturality follows trivially from the fomula.
\end{obs}

The next theorem gives the precise meaning of `acting in parallel': one can $\e$-thicken before taking the Reeb cosheaf, or $\e$-smooth after taking the Reeb cosheaf, and the result is the same.

\begin{thm}
\label{Thm:CommuteWithC}
We have $ \CC\TT_\e \simeq \SS_\e\CC$. That is, the functors $\CC\TT_\e$, $\SS_\e\CC$ are naturally isomorphic.
\end{thm}

The main part of the proof is understanding the relationship between inverse images of $f$ and $f_\e$.
Let $p: \X_\e = \X \times [-\e, \e] \to \X$ denote the projection onto the first factor.

\begin{lemma}
\label{Lem:Homotopy}
The map~$p$ restricts to a homotopy equivalence $f_\e^{-1}(I) \stackrel{\sim}{\longrightarrow} f^{-1}(I^\e)$ for each interval~$I$.
\end{lemma}

\begin{proof}
Let $p_I$ denote the restriction of~$p$ to $f_\e^{-1}(I)$. Then $p_I$ carries $f_\e^{-1}(I)$ into $f^{-1}(I^\e)$ because $f(x) + t \in I$ implies $f(x) \in I^\e$.

To define a homotopy inverse $q_I: f^{-1}(I^\e) \to f_\e^{-1}(I)$, we select a continuous function $\lambda: I^\e \to [-\e,\e]$ such that $s + \lambda(s) \in I$ for any $s \in I^\e$. For instance, if we write $I = (a,b)$ then $\lambda$ can be any continuous function on $(a-\e,b+\e)$ whose graph lies in the interior of the parallelogram shown here:
\[
\includegraphics[scale=0.5]{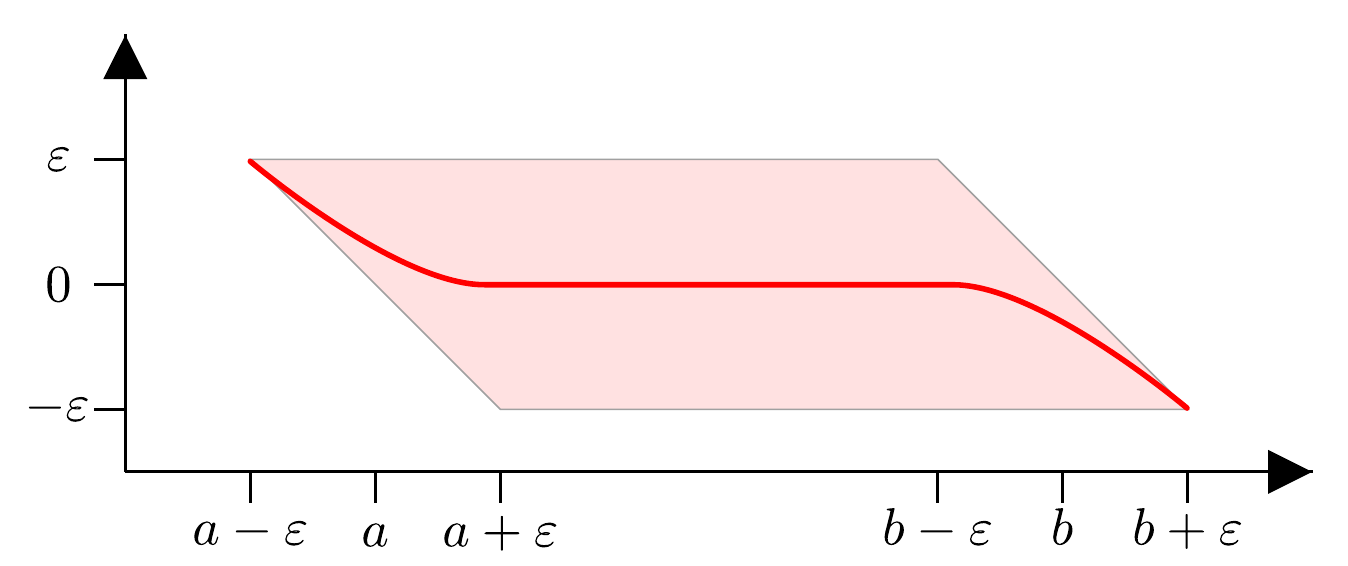}
\]
There is no problem choosing such a function for any given interval~$I$. We set $q_I(x) = (x,\lambda(f(x)))$. Then $q_I$ carries $f^{-1}(I^\e)$ into $f_\e^{-1}(I)$, since $f(x) \in I^\e$ implies $f(x) + \lambda(f(x)) \in I$ by the condition on~$\lambda$.

Certainly $p_I q_I$ is equal to the identity on $f^{-1}(I^\e)$. In the other direction, $q_I p_I(x,t) = (x, \lambda(f(x)))$ and this function is homotopic to the identity on $f_\e^{-1}(I)$ via linear interpolation in the $t$-coordinate. This works because for any fixed~$x$ the set $f_\e^{-1}(I)$ meets the fiber $\{x\} \times [-\e,\e]$ in an interval.
\end{proof}

\begin{proof}[Proof of Theorem~\ref{Thm:CommuteWithC}]
We will define a natural transformation $\rho: \CC\TT_\e \tO \SS_\e\CC$ and show that $\rho_f$ is an isomorphism for each object $f = (\X,f)$ in $\Rtop$.

Expanding the definitions,  $\CC\TT_\e(f)$ and $\SS_\e\CC(f)$ are pre-cosheaves which evaluate on intervals and morphisms as follows:
\begin{alignat*}{4}
&\big[ \CC \TT_\e(f) \big] (I) &&= \pi_0 f_\e^{-1}(I)
&&\big[ \CC \TT_\e(f) \big][I \subseteq J]
&&= \pi_0[ f_\e^{-1}(I) \subseteq f_\e^{-1}(J) ]
\\
&\big[ \SS_\e\CC(f) \big] (I) &&= \pi_0 f^{-1}(I^\e)
\qquad
&&\big[ \SS_\e\CC(f) \big][I \subseteq J]
&&= \pi_0[ f^{-1}(I^\e) \subseteq f^{-1}(J^\e) ]
\end{alignat*}
We define the natural transformation $\rho_f: \CC \TT_\e(f) \tO \SS_\e\CC(f)$ by the formula $(\rho_f)_I = \pi_0[p_I]$. Here $p_I$ is the map defined in Lemma~\ref{Lem:Homotopy}. Since it is a homotopy equivalence, it follows that $(\rho_f)_I$ is an isomorphism.
Applying $\pi_0$ to the left square of the commutative diagram
\[
\begin{diagram}
\dgARROWLENGTH=2em
\node{f_\e^{-1}(I)}
	\arrow{e}
	\arrow{s,l}{p_I}
\node{f_\e^{-1}(J)}
	\arrow{s,l}{p_J}
	\arrow{e}
\node{\X \times [-\e,\e]}
	\arrow{s,r}{p}
\\
\node{f^{-1}(I^\e)}
	\arrow{e}
\node{f^{-1}(J^\e)}
	\arrow{e}
\node{\X}
\end{diagram}
\]
we confirm that $\rho_f$ is a natural transformation; that is, a morphism of pre-cosheaves. Since each $(\rho_f)_I$ is an isomorphism it follows that $\rho_f$ is an isomorphism of pre-cosheaves.

To finish we must show that the family of pre-cosheaf isomorphisms $\rho = (\rho_f)$ is natural with respect to morphisms in $\Rtop$. In fact, for any morphism $\alpha: (\X,f) \to (\Y,g)$ we have a commutative diagram
\[
\begin{diagram}
\dgARROWLENGTH=2em
\node{f_\e^{-1}(I)}
	\arrow{e,t}{\alpha \times \mathbb{1}}
	\arrow{s,l}{p_I^f}
\node{g_\e^{-1}(I)}
	\arrow{s,r}{p_I^g}
\\
\node{f^{-1}(I^\e)}
	\arrow{e,t}{\alpha}
\node{g^{-1}(I^\e)}
\end{diagram}
\]
for each interval~$I$, to which we can apply $\pi_0$ to get the required naturality condition.
\end{proof}

The thickening functors $\TT_\e$ preserve constructibility. We give a simpler result first, since we can state its proof more briskly and it is all we need for the topological smoothing of Reeb graphs.

\begin{prop}
\label{prop:thick-graph}
If $(\X,f) \in \Reeb$ then $\TT_\e(\X,f) \in \Rtopc$.
\end{prop}

\begin{proof}
An $\R$-graph can be represented as a piecewise linear function~$f$ on a compact 1-dimensional polyhedron~$\X$. By construction, $\X_\e$ is a compact polyhedron and $f_\e$ is piecewise linear, so $\TT_\e(\X,f) = (\X_\e, f_\e)$ is a constructible $\R$-space.
\end{proof}

Here is the full result.

\begin{thm}
\label{thm:thick-R-space}
If $(\X,f) \in \Rtopc$, then $\TT_\e(\X,f) \in \Rtopc$.
%
\end{thm}

\begin{proof}
It will be helpful to reparametrize~$(\X_\e,f_\e)$.
\begin{figure}
\centering
 \includegraphics[width=.3\textwidth]{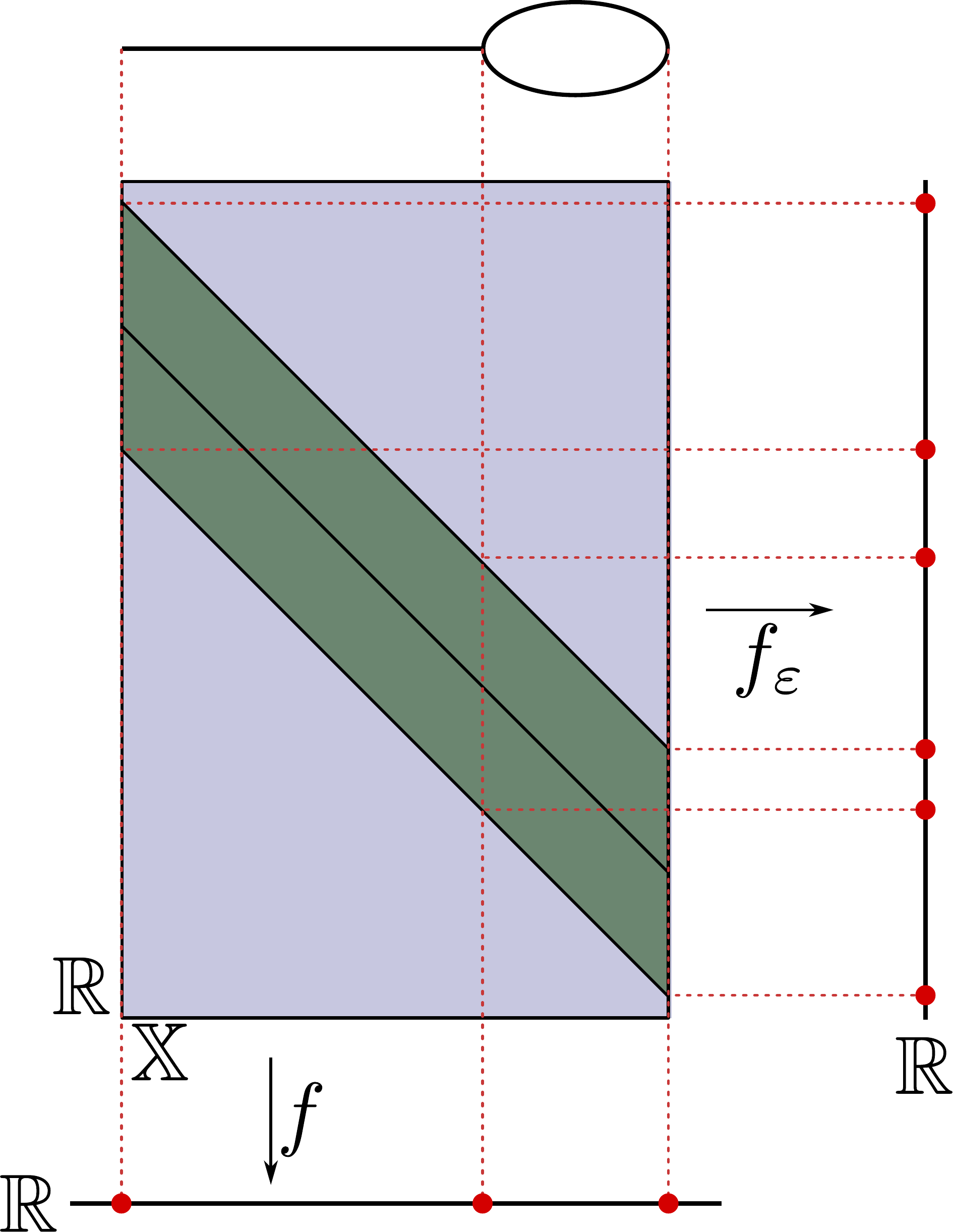}
 \caption[Critical points of a thickening]{The thickened space $\X_\e$ reparametrized as a subspace (green band) of $\X \times \R$ (tall rectangle). The map~$f_\e$ becomes the projection onto the factor~$\R$. In this example the three critical values of $(\X,f)$ give rise to six critical values of $(\X_\e,f_\e)$. The fibers at and between the critical values are illustrated in Figure~\ref{Fig:IntermediateSets}.
}
\label{Fig:Reparametrization}
\end{figure}
Consider the $\R$-space defined as follows:
\[
\tilde{\X}_\e
=
\{ (x,u) \in \X \times \R \mid  u - \e \leq f(x) \leq u + \e \},
\quad
\tilde{f}_\e(x,u) = u.
\]
We claim that $(\X_\e, f_\e)$ is isomorphic to $(\tilde{\X}_\e, \tilde{f}_\e)$. 
An inverse pair of morphisms defined as follows:
\begin{alignat*}{5}
(\X_\e,f_\e) &\to (\tilde{\X}_\e, \tilde{f}_\e);\;\;
&&(x,t) &&\mapsto (x, t + f(x))
\\
(\tilde{\X}_\e, \tilde{f}_\e) &\to (\X_\e,f_\e);\;\;
&&(x,u) &&\mapsto (x, u - f(x))
\end{alignat*}
See Figure~\ref{Fig:Reparametrization}.
For the rest of the proof, we drop the tildes and write $(\X_\e,f_\e)$ to mean $(\tilde{\X}_\e,\tilde{f}_\e)$.

\medskip
{\it Step 1: Compact locally path-connected fibers.}
Notice that (in the new coordinates) we have
\[
f_\e^{-1}(u) = f^{-1}[u-\e,u+\e] \times \{ u \}.
\]
Each point of $f^{-1}[u-\e,u+\e]$ has a neighborhood which looks like a cylinder on some~$\E_i$ (in the non-critical fibers) or a mapping cylinder to some~$\V_i$ (in the critical fibers). Since the $\V_i, \E_i$ are locally path-connected, the same is true for these neighborhoods. Thus each fiber is locally path-connected; and compact because $\X$ is compact.

\medskip
{\it Step 2: Critical set.}
We define $S_\e= \{ a+\e, a-\e \mid a \in S \}$.
These are precisely the values of~$u$ where one of the endpoints of $[u-\e,u+\e]$ meets the critical set~$S$.
Away from these values, we find that the fibers of~$\X_\e$ are locally constant in topological type. Write $S_\e = \{ b_0, b_1, \dots, b_m\}$ in increasing order.

\medskip
{\it Step 3: Critical fibers.}
We set $\W_k = f^{-1}[b_k-\e,b_k+\e]$ and note that $\W_k \times \{b_k\} = f_\e^{-1}(b_k)$.

\medskip
{\it Step 4: Non-critical fibers.}
We set $\F_k = f^{-1}(b)$ for some $b \in (b_k,b_{k+1})$. This fiber takes the following form. If $S$ meets $[b-\e,b+\e]$ in a nonempty set $\{a_i, a_{i+1}, \dots, a_j\}$ then
\[
\F_k
=
\E_{i-1} \times [b-\e, a_i]
\;\cup\;
f^{-1}[a_i, a_j]
\;\cup\;
\E_{j} \times [a_j,b+\e]
\]
If $S$ does not meet $[b-\e,b+\e]$ then simply $\F_k = \E_i \times [b-\e,b+\e]$ for some~$i$.

\medskip
{\it Step 5: Cylindrical structure maps.}
We define $\alpha_k: \F_k \times [b_k, b_{k+1}] \to \X_\e$ as follows.
First we define a map $\alpha_k^u: \F_k \to f^{-1}[u-\e,u+\e]$ for each $u \in [b_k,b_{k+1}]$.
If $S$ meets $[b-\e,b+\e]$ we use the following diagram:
\[
\begin{array}{rcccccccl}
\F_k
&=
& \E_{i-1} \times [b-\e, a_i]
& \cup
& f^{-1}[a_i, a_j]
& \cup
& \E_{j} \times [a_j,b+\e]
\\
&&
\text{\sc l}\big\downarrow
&&
\text{\sc m}\big\downarrow
&&
\text{\sc r}\big\downarrow
&&
\phantom{\Big\downarrow}
\\
&
& \E_{i-1} \times [u-\e, a_i]
& \cup
& f^{-1}[a_i, a_j]
& \cup
& \E_{j} \times [a_j,u+\e]
& \longrightarrow
& f^{-1}[u-\e,u+\e]
\end{array}
\]

The map \textsc{m} is the identity and each map \textsc{l} and~\textsc{r} is the homeomorphism defined by linearly stretching the second factor of the domain onto the second factor of the codomain. If $S$ does not meet $[b-\e,b+\e]$ then we use the diagram
\[
\begin{array}{ccccccc}
\F_k
& =
& \E_i \times [b-\e,b+\e]
& \stackrel{\textsc{m}}{\longrightarrow}
& \E_i \times [u-\e,u+\e]
& \longrightarrow
& f^{-1}[u-\e,u+\e]
\end{array}
\]
where \textsc{m} is the homeomorphism defined by linearly translating the second factor.

Then $\alpha_k(\xi,u) = (\alpha_k^u(\xi), u)$ is the required cylindrical structure map. It is continuous because the coefficients of the stretches or translations are continuous in~$u$.

Note that $\alpha_k^u$ is a homeomorphism when $u \in (b_k,b_{k+1})$. We interpret the two endpoint cases as attaching maps $\alpha_k^{b_k}: \F_k \to \W_k$ and $\alpha_k^{b_{k+1}}: \F_k \to \W_{k+1}$. These need not be homeomorphisms.

\medskip
{\it Step 6: Constructibility.}
We  now have maps from the spaces $\W_k \times \{b_k\}$ and $\F_k \times [b_k,b_{k+1}]$ to~$\X_\e$
which respect the attaching maps $\alpha_k^{b_k}, \alpha_k^{b_{k+1}}$. By the quotient principle this induces a continuous map from the constructible $\R$-space built from the $\W_k, \F_k$ and the corresponding attaching maps. This map is a bijection on each fiber and therefore a bijection. Since the domain is compact and the codomain is Hausdorff, the map is a homeomorphism.

\medskip
This completes the proof that $(\X_\e, f_\e)$ is constructible.
See Figure~\ref{Fig:IntermediateSets} for an example.
\end{proof}

\begin{figure}
\centerline{
\includegraphics[width=0.6\textwidth]{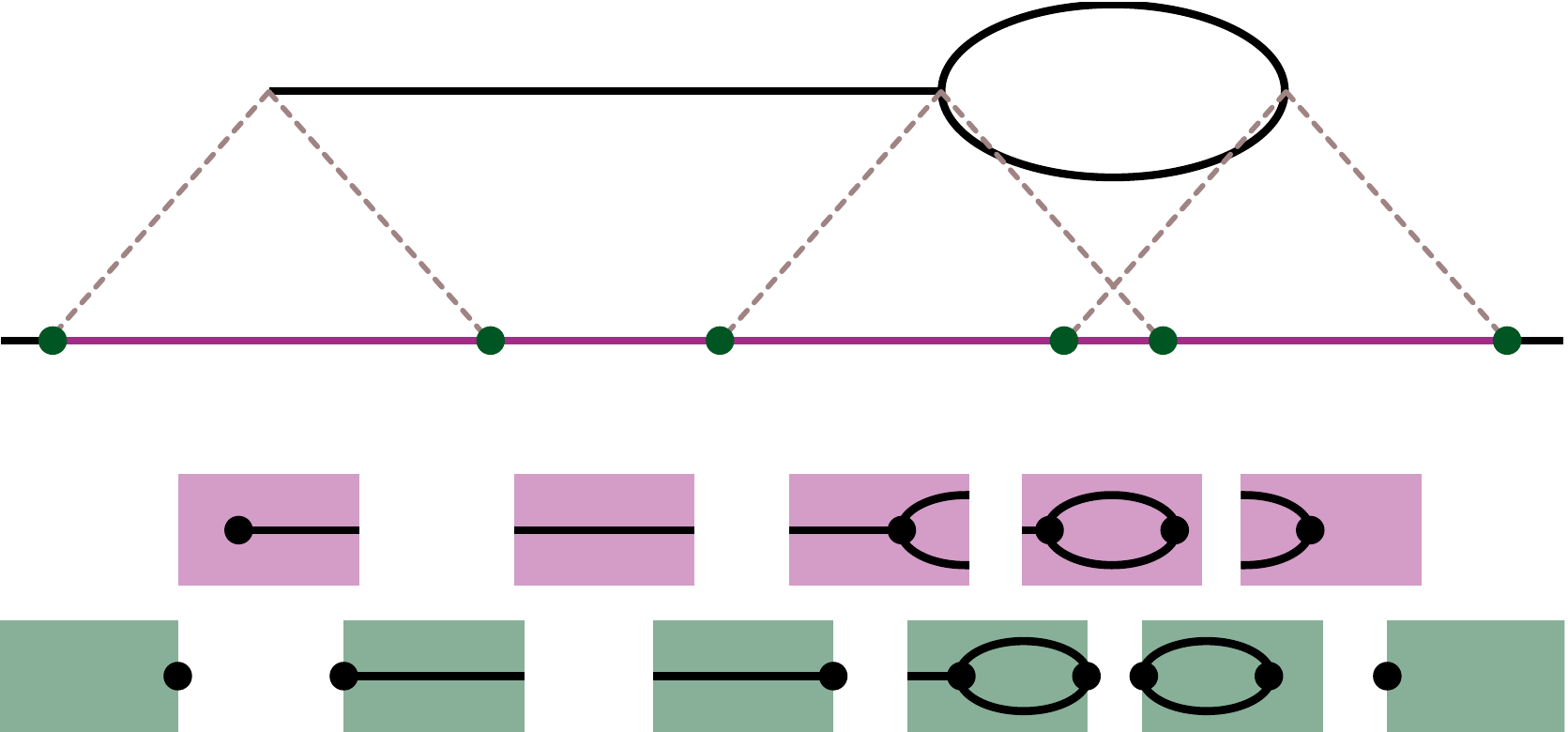}
}
 \caption[Critical and non-critical fibers of a thickening]{The critical fibers $\W_k$ (green) and non-critical fibers $\F_k$ (purple) of the thickened $\R$-space $(\X_\e,f_\e)$ in the example of Figure~\ref{Fig:Reparametrization}. The fibers are homeomorphic to interlevelsets $f^{-1}[a-\e,a-\e]$.
}
 \label{Fig:IntermediateSets}
\end{figure}

\begin{cor}
It follows from the proof that if $f$ has critical set~$S$ then $f_\e$ has critical set $S_e = (S-\e) \cup (S + \e) = \{ a-\e, a+\e \mid a \in S\}$.
\end{cor}

\begin{rmk}
The thickened space $(\X_\e,f_\e)$ can be thought of as giving a natural topology on the family of `sliding windows' on $(\X,f)$ of width~$2\e$.
\end{rmk}

\subsection{Topological smoothing of $\R$-graphs}
\label{Sect:smooth-Reeb}

We are now in a position to define a semigroup $(\UU_\e)$ of `topological smoothing' functors on $\R$-graphs. We then use these functors to reinterpret the Reeb distance in a more purely geometric way.
There are two reasonable ways to define this semigroup:
\begin{itemize}
\item
Use the thickening functors $(\TT_\e)$ followed by a projection onto $\Reeb$.

\item
Transfer the smoothing functors $(\SS_\e)$ to~$\Reeb$ using the equivalence of categories.
\end{itemize}
%
%
We favour the first method, which gives the following explicit definition, for all $\e \geq 0$:
\begin{defn}
Define the \emph{Reeb smoothing functor} $\UU_\e : \Reeb \to \Reeb$ by $\UU_\e = \RR\TT_\e$.
\end{defn}

Given an $\R$-graph $f = (\X,f)$, it follows that the fiber of $\UU_\e f$ over $t \in \R$ can be identified with the set of connected components of $f^{-1}[t-\e,t+\e]$. If the vertices of the original graph occur over a critical set~$S$, then the vertices of the smoothed graph occur over the set $S_\e = (S + \e) \cup (S - \e)$.
These facts follow from Theorem~\ref{thm:thick-R-space} and its Corollary.

\begin{ex}
To better understand this construction, consider the examples of Figure \ref{Fig:SmoothedReebAll}.  
The initial $\R$-graph $\X$ is given in column $(a)$ with the function $f$ implied by the height.  
In order to visualize the space $\X \times [-\e,\e]$ in column $(b)$, $\X$ is redrawn with a $[-\e,\e]$ interval added at each point.
Since these intervals are drawn vertically, the function $f_\e$ can still be visualized  as the height function of this new space.
The Reeb graph of this space is overlaid in column $(c)$ and drawn by itself in column $(d)$.

\begin{itemize}
\item[(1)]
Here the $\R$-graph is a line. It gets stretched by $\e$ in both directions.

\item[(2)]
The up-fork in this example gets pushed up by~$\e$. (A down-fork would get pushed down by~$\e$.)

\item[(3)]
This example has a loop with height $(b-a) \leq 2\e$. After thickening, every levelset has only one connected component so the resulting Reeb graph has gotten rid of the loop entirely.

\item[(4)]
This is a more complicated mix of the ingredients above. Note that the height of the loop shrinks by $2\e$. There is interesting behavior on the right side where a down-fork interacts with an up-vee.

\end{itemize}
\end{ex}

It is easiest to access the properties of $(\UU_\e)$ by comparing these functors with $(\SS_\e)$.

\begin{prop}
\label{prop:U=S}
The functors $\CC''\UU_\e$ and $\SS_\e \CC'': \Reeb \to \Cshc$ are naturally isomorphic.
\end{prop}

\begin{proof}
We have $\CC'' \RR \TT_\e \simeq \CC' \TT_\e \simeq \SS_\e \CC''$ by Theorems \ref{Thm:CommutativeTriangle} and~\ref{Thm:CommuteWithC}.
\end{proof}

This implies, in particular, that the functors $(\hat\UU_\e = \DD \SS_\e \CC'')$ suggested by the second method above are naturally isomorphic to the functors $(\UU_\e)$. We prefer the first method because it is more geometric and because the inverse functor~$\DD$ used by the second method has not been defined explicitly.

\begin{figure}
\centering
 \includegraphics[scale=.7]{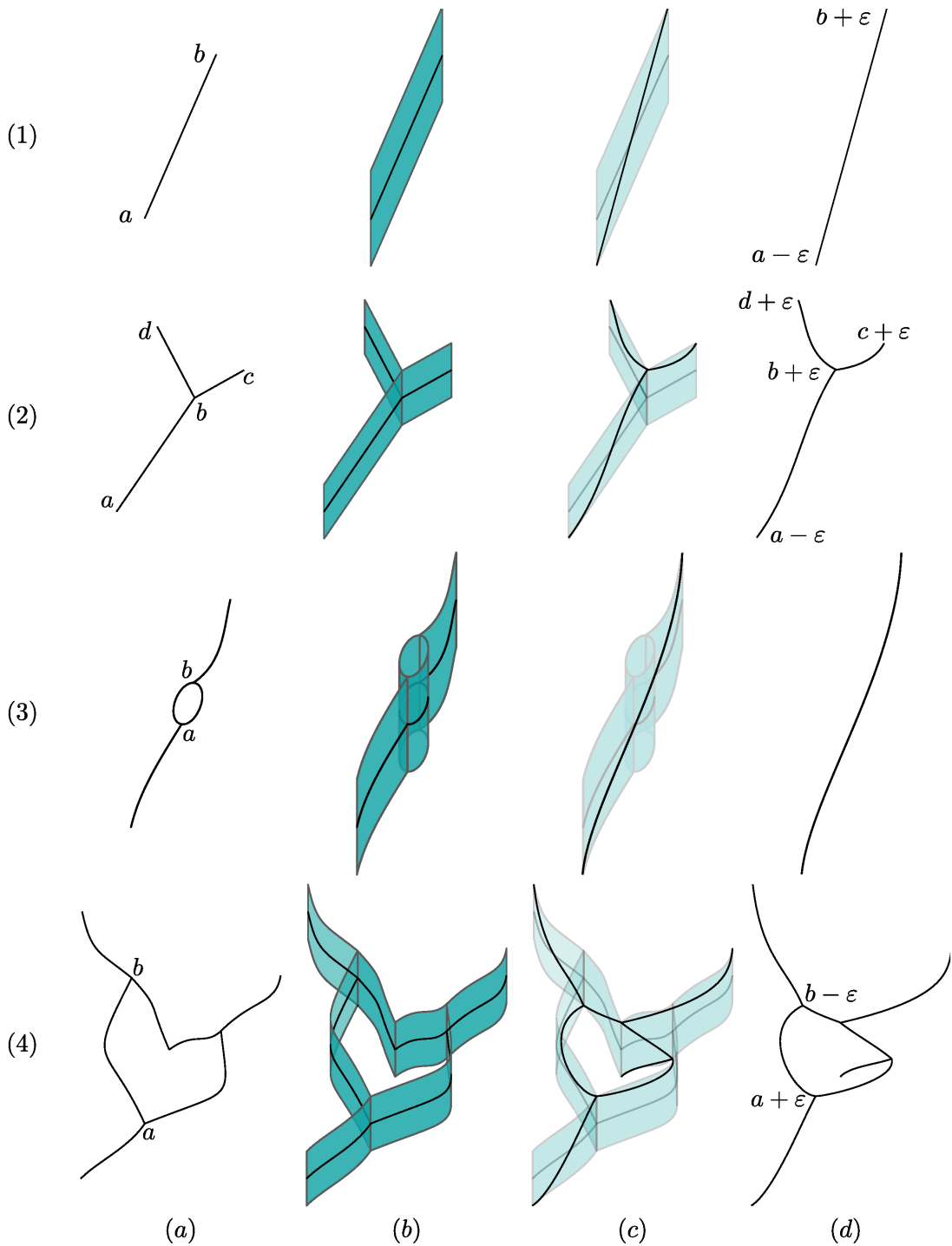}
 \caption[Smoothing of $\R$-graphs]{Examples of smoothing.  Column $(a)$ represents the original $\R$-graph~$\X$ with function $f$ shown by vertical height. Column $(b)$ shows the thickened space $\X \times [-\e, \e]$ with function $f_\e$ still shown by the vertical height. Column $(c)$ shows the Reeb graph of $(\X \times [-\e,\e], f_\e)$ superimposed on $\X \times [-\e, \e]$. Column $(d)$ shows this new $\R$-graph by itself. %
}
 \label{Fig:SmoothedReebAll}
\end{figure}

\begin{obs}
\label{obs:semigroup-g}
The family of functors $(\UU_\e)$ form a semigroup of contraction endofunctors (in $\Reeb$), in the sense that:
\begin{itemize}
\item
$\UU_0 \simeq \onefunc_\Reeb$ and $\UU_{\e_1+\e_2} \simeq \UU_{\e_1} \UU_{\e_2}$ for all $\e_1,\e_2 \geq 0$;

\item
$\rdist(\UU_\e(f), \UU_\e(g)) \leq \rdist(f,g)$ for all $f,g \in \Reeb$ and $\e \geq 0$.

\end{itemize}
These assertions follow immediately from the corresponding assertions (Propositions \ref{prop:smoothing-semigroup} and~\ref{prop:smoothing-contraction}) for the family of endofunctors $(\SS_\e)$ of $\Cshc$, thanks to Proposition~\ref{prop:U=S}.
For the semigroup property, we have to replace `$=$' with `$\simeq$' since that is all we can deduce from an equivalence.
\end{obs}

\begin{obs}
\label{obs:smooth-2c}
Each $\UU_\e$ is a pointed endofunctor of $\Reeb$: there is a family $\zeta^\e = (\zeta^\e_f)$ of maps
\[
\zeta^\e_f: f \to \UU_\e(f)
\]
from each $\R$-graph to its $\e$-smoothing, which constitute a natural transformation $\zeta^\e: \onefunc_\Reeb \tO \UU_\e$.
For a given graph $f = (\X,f)$, the map $\zeta^\e_f$ is the composite
\[
\begin{diagram}
\node{\X}
	\arrow{e}
\node{\X_\e}
	\arrow{e}
\node{\X_\e/\sim}
\end{diagram}
\]
where the first map is the inclusion of~$\X$ as the zero-section of $\X_\e$ and the second map is the Reeb quotient.
That is to say, $\zeta^\e$ is the composite of the natural transformations $\tau^\e$ and~$\rho_{f_\e}$ of Observations \ref{obs:thick} and~\ref{obs:Reeb-rho}.
Geometrically, the map sends each point of $f^{-1}(t)$ to the connected component of $f^{-1}[t-\e,t+\e]$ that it belongs to.
\end{obs}

We finish this section by showing that $\zeta^\e: \onefunc_{\Reeb} \tO \UU_\e$ corresponds exactly to $\sigma^\e: \onefunc_{\Cshc} \tO \SS_\e$.

\begin{prop}
\label{prop:pointed-equivalence}
For $f \in \Reeb$ and $\e \geq 0$, the right-hand square in the following figure
\[
\xymatrix
@=8ex
{
f
	\ar@{->}[d]_(0.45){\zeta} 
\\
\UU_\e(f)
}
\qquad
\qquad
\xymatrix
@=8ex
{
\CC(f)
	\ar@{->}[d]_(0.45){\CC[\zeta]}
	\ar@{=}[r]
&
\CC(f)
	\ar@{->}[d]^(0.45){\sigma}
\\
\CC \UU_\e (f)
	\ar@{<-}[r]^{\simeq}
&
\SS_\e \CC(f)
}
\]
commutes. Here $\zeta = \zeta^\e_f$ and $\sigma = \sigma^\e_{\CC(f)}$. The left-hand side of the square is obtained by applying $\CC$ to the diagram on the left, and the map at the bottom of the square is the isomorphism of Proposition~\ref{prop:U=S}.
\end{prop}

\begin{proof}
We need to verify that the square (of small functors and natural transformations) commutes when evaluated at an arbitrary interval~$I$. The result of this evaluation is the left-hand square below.
Unpacking the definitions, we find that it is the image under $\pi_0$ of the right-hand square below:
\[
\renewcommand{\labelstyle}{\textstyle}
\xymatrix
@=8ex
{
\pi_0 f^{-1}(I)
	\ar@{->}[d]_(0.45){\CC[\zeta]_I}
	\ar@{=}[r]
&
\pi_0 f^{-1}(I)
	\ar@{->}[d]^(0.45){\sigma_I}
\\
\pi_0 f_\e^{-1}(I)
	\ar@{<-}[r]^{\simeq}
&
\pi_0 f^{-1}(I^\e)
}
\qquad
\raisebox{-5.5ex}{$\stackrel{\pi_0}{\longleftarrow}$}
\qquad
\xymatrix
@=8ex
{
f^{-1}(I)
	\ar@{->}[d]_(0.45){\textsc{l}}
	\ar@{=}[r]
&
f^{-1}(I)
	\ar@{->}[d]^(0.45){\textsc{r}}
\\
f_\e^{-1}(I)
	\ar@{<-}[r]^{\textsc{h}}
&
f^{-1}(I^\e)}
\]
The three labelled maps are inclusions. The left map~\textsc{l} is the inclusion of $f^{-1}(I)$ as the zero-section of $f_\e^{-1}(I)$. The right map~\textsc{r} is the inclusion of $f^{-1}(I)$ as a subset of $f^{-1}(I^\e)$. The horizontal map~\textsc{h} is the homotopy equivalence defined in Lemma~\ref{Lem:Homotopy} in terms of the auxiliary function~$\lambda$.

It is enough to show that this right-hand square commutes up to homotopy. Indeed, this is the case: the maps \textsc{l} and~\textsc{hr} are homotopic via a fiberwise straight-line homotopy.
This may be understood by contemplating the figures
\medskip
\begin{center}
\raisebox{3ex}{\textsc{l:}}
\includegraphics[height=8ex]{thetagram-3}
\qquad
\raisebox{3ex}{\textsc{hr:}}
\includegraphics[height=8ex]{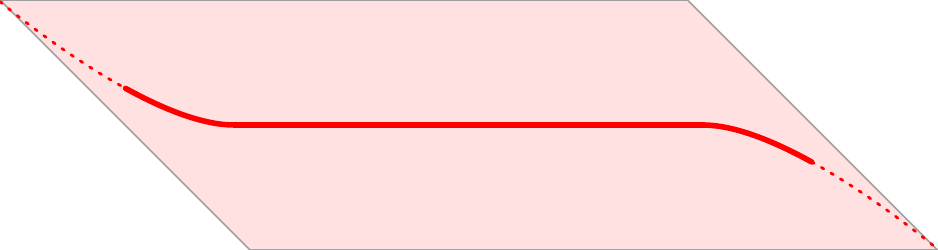}
\end{center}
which schematically represent the two homotopic maps.
\end{proof}

\subsection{Interleaving of $\R$-graphs}
\label{Sect:interleaving-Reeb}

We are now in a position to interpret the interleaving distance between $\R$-graphs geometrically. The original version (Definitions \ref{def:interleaving} and~\ref{def:interleaving-dist}) asks us to compare two cosheaves. The discussion in Section~\ref{Sect:smooth-Reeb} allows us to do the comparison directly on the graphs themselves.
Let $f = (\X,f)$ and $g = (\Y,g)$ be $\R$-graphs. An $\e$-interleaving of their Reeb cosheaves $\Ffunc, \Gfunc$ is described by the diagrams in~\eqref{eq:interleaving-cat}. We interpret this in the geometric category.

\begin{defn}
\label{def:interleaving-g}
Two $\R$-graphs $f,g$ are {\bf $\e$-interleaved} if there exist maps
\[
\alpha: f \to \UU_\e g
\quad \text{and} \quad
\beta: g \to \UU_\e f
\]
such that the diagrams
\begin{equation}
\label{eq:interleaving-graph}
\renewcommand{\labelstyle}{\textstyle}
\raisebox{9.5ex}{
\xymatrix
@=6ex
{
f
	\ar@{->}[dr]^(0.45){\alpha}
	\ar@{->}[dd]_(0.45){\zeta^{2\e}_{f}}
\\
&
\UU_\e g
	\ar@{->}[dl]^(0.45){\beta^{\e}_{2\e}} 
\\
\UU_{2\e}f
}}
\qquad
{\text{and}}
\qquad
\raisebox{9.5ex}{
\xymatrix
@=6ex
{
&g
	\ar@{->}[dl]_(0.45){\beta}
	\ar@{->}[dd]^(0.45){\zeta^{2\e}_{g}}
\\
\UU_\e f
	\ar@{->}[dr]_(0.45){\alpha^\e_{2\e}} 
\\
&
\UU_{2\e}g
&
}
}
\end{equation}
%
commute. Here $\alpha^\e_{2\e}$ and $\beta^\e_{2\e}$ are the composites
\[
\renewcommand{\labelstyle}{\textstyle}
\xymatrix@C=6em@R=2ex{
\alpha^\e_{2\e} : 
\UU_\e f
	\ar[r]^{\UU_\e [\alpha]}
&
\UU_\e \UU_\e g
	\ar[r]^{\simeq}
&
\UU_{2\e} g
\\
\beta^\e_{2\e} : 
\UU_\e g
	\ar[r]^{\UU_\e [\beta]}
&
\UU_\e \UU_\e f
	\ar[r]^{\simeq}
&
\UU_{2\e} f
}
\]
where the right-hand maps are the natural isomorphisms of Observation~\ref{obs:semigroup-g}.
Through the equivalence of categories $\CC''$ and the results of the previous section, it follows that $f,g$ are $\e$-interleaved as $\R$-graphs if and only if $\Ffunc, \Gfunc$ are $\e$-interleaved as cosheaves.
\end{defn}

\begin{rmk}
Geometrically, the map $\alpha^\e_{2\e}$ acts as follows. Each point $x \in\UU_\e f$ corresponds to a connected component of some $f^{-1}[t-\e,t+\e]$. The points in that component are carried by~$\alpha$ to connected components of various $g^{-1}[s-\e,s+\e]$ where $s \in [t-\e,t+\e]$. By continuity, these components are all contained in a single connected component of $g^{-1}[t-2\e,t+2\e]$. It is this component that defines $\alpha^e_{2\e}(x) \in \UU_{2\e} g$. The map $\beta^\e_{2\e}$ acts in similar fashion.
\end{rmk}

The \emph{Reeb interleaving distance} between two $\R$-graphs 
is, finally, the infimum over values of~$\e$ for which there exists an $\e$-interleaving between the $\R$-graphs.

\section{Algorithms}

Our main goal in this section is to describe an algorithm for constructing the $\e$-smoothing of a given Reeb graph, as well as the canonical map from the graph to its smoothing. These are necessary ingredients for working with the interleaving distance.
This can be achieved in polynomial time.
Calculating the interleaving distance between two Reeb graphs is not so easy: in general it is graph-isomorphism-hard. We can, at least, recognize interleavings in polynomial time. We will discuss these matters in the later subsections.

\bigskip
Here is the set-up.
Let $f = (\X,f)$ be an $\R$-graph with critical set $S = \{a_0,\cdots,a_n\}$.  
In computational terms, this is just a graph $\X$ with function values associated to each vertex.  
Implicitly, we assume that the function value on the edges is a strictly monotone function with max and min determined by the function values at the vertices.

We wish to compute the smoothed Reeb graph $\UU_\e(f)$.
To do this, one might naively  build a larger complex $\X \times [-\e, \e]$ with function $f_\e$ as in column $(b)$ of Figure \ref{Fig:SmoothedReebAll} and then run any standard algorithm to compute its Reeb graph.
This new complex will have one edge and two vertices for every original vertex, and three edges and two faces for every edge, so for a graph with $m$ edges and $n$ vertices, the new complex has $O(m+n)$ total simplices.
Hence, the new Reeb graph can be computed in time $O((m+n)\log(m+n))$ in expectation using \cite{Harvey2010} or deterministically using \cite{Parsa2012}.

However, this method does not make use of the particular structure of the smoothing procedure. If we do exploit that structure, we can modify the algorithm of \cite{Parsa2012} to compute $\UU_\e(f)$ while  running in $O(m\log(m+n))$ time.  

\subsection{The smoothing algorithm}

Parsa's algorithm for computing the Reeb graph of an arbitrary simplicial complex in \cite{Parsa2012} is a sweep algorithm which keeps a data structure to represent the connected components of $f\inv(r)$ as the value $r$ increases in order to determine how the connected components should be attached.  
Let $f:\X \to \R$ be the original Reeb graph and $f_\e:\X\times[-\e,\e] \to \R$ the thickened Reeb graph.
Let $v_1,\cdots,v_n$ be the critical vertices of $f$ sorted so that $f(v_1) < f(v_2) < \cdots < f(v_n)$.
To simplify the explanation, we will assume general position, by which we mean $f(v_i) \neq f(v_j)$ and $f(v_i) \pm \e \neq f(v_j) \pm \e$ for any $i \neq j$.
Let $b_1< b_2 < \cdots b_k$ be the sorted values $\{f(v_i) \pm \e\}$; thus $\{b_i\}$ contains the critical values of $f_\e$.
For the sake of notation, let $[t]^\e$ denote the interval of width $2\e$ centered at $t$; that is, $[t]^\e = [t-\e, t+\e]$.
Also, we will say an edge $e = (v,w)$ with $f(v)<f(w)$ starts at $v$ and ends at $w$.

The main change here in order to instead compute the smoothed Reeb graph is the structure used to represent $f_\e \inv(t)$.  
Essentially, using Lemma \ref{Lem:Homotopy}, we know that we can work equivalently with $f_\e \inv (t)$ or $f\inv(t-\e,t+\e) = f\inv([t]^\e)$.
Thus, rather than dealing with the larger simplicial complex, we determine our connected components using a graph which keeps track of edges and vertices within that range.
We will use $H$ for the graph data structure which represents the connected components of  $f\inv([t]^\e)$ for the current value of~$t$. It should be noted that there are three main operations required for $H$: finding the particular connected component of either an edge or a vertex, merging two connected components, and splitting two connected components.
This structure and methods will be throughly discussed in the next section, and the full sweep algorithm will be explained in the section after that.

\subsubsection{Maintenance of Level Set Representation}

As we are working combinatorially, we will consider $f\inv([t]^\e)$ as a subgraph consisting of  all vertices with function value in $[t]^\e$ along with all edges attached to these vertices.
We adopt the convention that we can have `half edges' in this subgraph which occurs when an edge is attached to one vertex inside the interval and one outside.
In order to minimize any confusion as we pass back and forth between thinking of things topologically and combinatorially, let us represent the topological space $f\inv([t]^\e)$ by a graph $\overline{H}_t$ defined as follows:
\begin{itemize}
\item
every vertex in $f\inv([t]^\e)$ is represented by a vertex in $\overline{H}_t$;

\item
every edge whose interior meets $f\inv([t]^\e)$ is represented by a vertex in $\overline{H}_t$;

\item
every incidence between a vertex and an edge in $f\inv([t]^\e)$ is represented by an edge in $\overline{H}_t$.
\end{itemize}
Notice that some edges of the original graph may only partially meet $f\inv([t]^\e)$. Thus, every edge that meets $f\inv([t]^\e)$ can have 2, 1 or~0 vertices in $f\inv([t]^\e)$. The vertices of $\overline{H}_t$ that represent edges will therefore have degree 2, 1 or~0.

\begin{rmk}
The graph $\overline{H}_t$ is the \emph{derived complex}, or \emph{barycentric subdivision}, of the part of the Reeb graph contained in the window $f^{-1}([t]^\e)$. A similar construction is implied in Figure~\ref{Fig:CosheafConditionExample}.
\end{rmk}


\begin{algorithm}
\caption{UpdateH$(v,b)$}\label{Alg:UpdateH}
\begin{algorithmic}
\If {$b < f(v)$}
\State Add a vertex to $H$ for $v$
\For {$e$ below $v$ in $\X$}
\State $H.$insert$(e,v)$
\EndFor
\For {$e$ above $v$ in $\X$}
\State Add a vertex to $H$ for $e$
\State $H.$insert$(v,e)$
\EndFor
\Else
\For {$e$ below $v$ in $\X$}
\State $H.$delete$(e,v)$
\State Remove vertex $e \in V(H)$ 
\EndFor
\For {$e$ above $v$ in $\X$}
\State $H.$delete$(v,e)$
\EndFor
\State Remove vertex $v \in V(H)$
\EndIf
\end{algorithmic}
\end{algorithm}

As in \cite{Parsa2012}, we need to quickly determine the connected components of $\overline {H_t}$.  
This is done by solving the dynamic graph connectivity problem, where we store a rooted spanning forest of the graph $\overline{H}_t$; this forest is the graph notated $H$.
Since $H$ is a subset of $\overline{H}_t$, vertices in $H$ are associated to either a vertex or edge in the original graph $\X$.  

As the value of $t$ increases past a critical value $b = f(v)\pm \e$, we need to be able to update $H$ to continue to reflect the connected components.
Algorithm \ref{Alg:UpdateH}  outlines the procedure for this, which will be defined as UpdateH$(v,b)$.
If $b = f(v)-\e$, and thus $b<f(v)$, raising the value of $t$ from $b-\delta$ to $b+\delta$ requires adding the vertex $v$ to $H$, attaching all the edges which end at $v$, and starting the edges which emanate from $v$.  
On the other hand, if $b= f(v)+\e$, raising the value of $t$ requires removing $v$ from $H$, deleting the edges which end at it, and freeing the bottom of the edges above it.
Note that vertices in $H$ are only deleted after all attached edges are removed.

We can now look at how to implement insert, delete, and find in this graph.
In order to assure that we are not spending extra time adding and deleting edges, we will give each edge in $H$ a weight $\omega$ equal to the time that it will be deleted from $H$.
The purpose of these weights can be seen in the example of Fig.~\ref{F:WeightRequirement}.
A portion of the Reeb graph is at the left, and two choices for minimum spanning trees of the shaded region, $f\inv(I)$, are given.  
Circular vertices correspond to vertices in $f\inv(I)$, and square vertices correspond to edges.
If the leftmost MST is chosen, the only edits required as the interval $I$ moves up is to delete edges from the MST as they are removed from $f\inv(I)$.
If the right tree is used, the MST will require more complicated edits when the interval moves past the bottom vertex.
Note that since every edge in $H$ (or $\overline{H}_t$) has one endpoint corresponding to a vertex and the other to an edge, we can define $\omega(e)$ to be $f(v)$ for $v$ the vertex.
We will maintain that our minimum spanning tree utilizes edges with higher weights when possible.

\begin{figure}
\centering
 \includegraphics{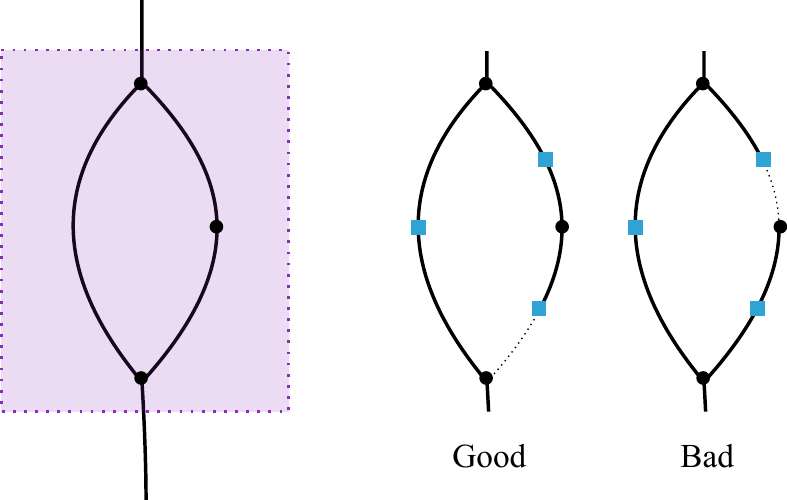}
 \caption[Choice of spanning tree in smoothing algorithm]{For the portion $f\inv(I)$ of the Reeb graph shown at left, two possible minimum spanning trees, $H$, are shown at right. In the left version, edges will only need to be removed from $H$ as they are removed from $f\inv(I)$ as $I$ is moved up. In the far right version, edges will need to be added and removed as $I$ is moved up.  Thus, we prefer the example on the left to decrease the number of edits done to $H$. The edge weight in the algorithm is used to avoid the example on the right.}
 \label{F:WeightRequirement}
\end{figure}


The operations needed on $H$ (where $x$, $x_1$, and $x_2$ are all vertices in $H$) in order to implement find, insert, and delete are:
\begin{itemize}
 \item parent$(x)$: return the parent of $x$, or null if $x$ is the root.
 \item root$(x)$: return the root of the tree containing $x $.
 \item link$(x_1,x_2,w)$: add an edge between $x_1$ and $x_2$ with weight $w$.
 \item cut$(x_1,x_2)$: delete the edge between $x_1$ and $x_2$.
 \item minWeight$(x)$: return a node with minimum weight edge to its parent on the path from $x$ to the root of its tree .
 \item evert$(x)$: make $x$ the root of its tree.
\end{itemize}
The dynamic graph connectivity problem is well studied with methods that include RC-trees \cite{Acar2004}, link-cut trees \cite{Sleator1983,Sleator1985}, top trees \cite{Tarjan2005,Alstrup2005}, and sparsification \cite{Eppstein1997}.
These can be implemented in order to perform the above operations in $O(\log n)$ worst case, amortized, or expected time where $n$ is the number of vertices in the graph; thus we leave the extended discussion of these specifics to the references.

Given the above methods, we can implement the three major operations as follows:
\begin{itemize}
 \item find$(x)$:
 \begin{algorithmic}
  \State \Return root$(x)$
 \end{algorithmic}
 
 \item insert$(e)$:
 \begin{algorithmic}
  \State $w = \omega(e)$ for edge $e = x_1x_2$
  \If { root$(x_1)=$root$(x_2)$}
  \State evert$(x_1)$
  \State $x',w'=$minWeight$(x_2)$
  \If {$w'<w$}
  \State cut$(x',$parent$(x'))$
  \State link$(x_1,x_2,w)$
  \EndIf
  \Else 
  \State {link$(x_1,x_2,w)$}
  \EndIf
 \end{algorithmic}
 
 \item delete$(e)$:
 \begin{algorithmic}
  \If {$e \in E(H)$}
  \State cut$(e)$
  \EndIf
 \end{algorithmic}

\end{itemize}
Notice that since the first set of operations can be implemented in $O(\log n)$ time, find, insert, and delete can be as well.

\subsubsection{Full Algorithm}

\begin{algorithm}
\caption{Sweep algorithm}\label{Alg:Sweep}
\begin{algorithmic}
\State $\textrm{Set }H\textrm{ to be an empty graph}$
\For {$i = 1, \cdots, k$ where $b_i = f(v_j)\pm\e$  }
\State $L_c = \textrm{LowerComps}(v_j,b_i)$
\State $\textrm{UpdateH}(v,b)$
\State $U_c = \textrm{UpperComps}(v_j,b_i)$
\State $\textrm{UpdateReebGraph}(L_c,U_c)$
\EndFor
\end{algorithmic}
\end{algorithm}
The pseudocode for the sweep algorithm is given in Alg.~\ref{Alg:Sweep}.  
Here, we work our way up the potential critical values $b_i$, keeping track of the change in $\overline{H}_t$ and using this to build the smoothed graph $g:\Y\to \R$.
For any noncritical $t$, the components of $\overline{H}_t$ are associated to an edge in $\Y$ whose lower vertex $v$ and upper vertex $w$ satisfy $g(v)<t<g(w)$.
We will keep track of this association by pointing the representative of a component in $H$ to the lower vertex of its associated edge in~$\Y$.

At the beginning of a step, $H$ represents the connected components for $\overline{H}_t$ with $t = b_i-\delta$ for a sufficiently small $\delta$.
First, we find the components in $H$ that could be impacted by the addition or deletion of $v_j$ and associated edges using the LowerComps$(v_j,b_i)$ subprocess, Alg.~\ref{Alg:LowerComps}.  
These are stored in $L_c$.
Note that these are exactly the edges which end at $\nu$, where $\nu$ is the vertex added at function value $b_i$ in the smoothed Reeb graph.
\begin{algorithm}
\caption{LowerComps$(v,b)$}\label{Alg:LowerComps}
\begin{algorithmic}
\If {$b<f(v)$}
\For {edges $e$ ending at $v$}
\State $c = H.\textrm{find}(e)$
\If {$c$ is not marked}
\State $L_c.$add$(c)$
\State Mark $c$ as listed
\EndIf
\EndFor
\Else
\State $L_c = H.$find$(v)$
\EndIf
\end{algorithmic}
\end{algorithm}

Then the $H$ graph is updated so that it now represents the connected components for $t = b_i+\delta$ using the UpdateH$(v,b)$ subprocess, Alg.~\ref{Alg:UpdateH} discussed in the previous section.
The components in the new $H$ are determined using UpperComps$(v_j,b_i)$ (which is symmetric to LowerComps$(v_j,b_i)$ and therefore not repeated here) and stored as $U_c$.  
These components are the edges which start at $\nu$.  

Finally, to update $\Y$, we use UpdateReebGraph (Alg.~\ref{Alg:UpdateReebGraph}) to add a new vertex $\nu$ to the graph.
An edge is added for each lower component in $L_c$ which starts at the associated start vertex and ends at $\nu$.
Then each component in $U_c$ is assigned $\nu$ as the start vertex.
\begin{algorithm}
\caption{UpdateReebGraph$(L_c,U_c)$}\label{Alg:UpdateReebGraph}
\begin{algorithmic}
\If{$\#|L_c| = \#|U_c| = 1$}
\State \Return
\Else
\State $\textrm{Create a new node }\nu \textrm{ in Reeb graph}$
\State $\textrm{Assign that node to all }c \in U_c$
\State $\textrm{Add an edge between }\nu\textrm{ and }v_c \textrm{ for all } c \in L_c$
\EndIf
\end{algorithmic}
\end{algorithm}

\subsection{Analysis of the smoothing algorithm}

We show that the overall running time is linear in the total number of simplices of the original Reeb graph. 
In particular, for a graph with $n$ vertices and $m$ edges, the running time is $O(m\log(n+m))$.
First, note that the number of vertices in $H$ is at most $n+m$, thus the time for any find, insert, or delete is $O(\log (m+n))$.
Every edge of the original graph is used once for LowerComps and once for UpperComps.  
Since each of these utilizes one find operation, the total time spent in them is $O(m\log(n+m))$.
Likewise, edges are added to $H$ twice for each original edge, and these edges are each deleted.  
Thus the total time spent in UpdateH is also $O(m\log(n+m))$.
Thus, the smoothing algorithm runs in $O(m\log(n+m))$ time.

\subsection{Morphisms between graphs with different critical sets}
\label{subsec:morphisms-diff-S}

In recording morphisms, it is helpful to notice that we can get away with storing less information than is indicated by Proposition~\ref{Prop:MapData}. In particular we can to write down maps between Reeb graphs without needing to subdivide the graphs to have a common set of critical values.  

Consider a map $\phi:(\X,f) \to (\Y,g)$ given by data $\phi_i^V:V_i^\X \to V_i^\Y$ and $\phi_i^E:E_i^\X \to E_i^\Y$ as specified in in Proposition~\ref{Prop:MapData}.
Say $(\Y,g)$ is a subdivision of $(\Y',g)$.  
Thus $|\Y| \iso |\Y'|$ but $\Y'$ may have some vertices with up and down degree both 1 which are not in $\Y$.  
We abuse notation by calling both maps $g$.
Since we assume that $g$ is monotone on any edge, if there is a vertex in $\X$ which maps to $v \in V_i^\Y$ under $\phi$, the same information is stored if we say that it maps to the edge $e \in E_i^{\Y'}$.
Thus, for the non-subdivided version of $\phi$, we instead have a map $\phi_i^V: V_i^\X \to V_i^{\Y'} \cup E_i^{\Y'}$.
In order to get rid of confusion with indices, we will denote this map 
\begin{equation*}
 \phi^V:V^\X \to V^{\Y'} \cup E^{\Y'}
\end{equation*}
while remembering that a vertex $v  \in \X$ will map to either a vertex $w \in V^{\Y'}$ with $f(v) = g(w)$ or to an edge $e = (u,w) \in E^{\Y'}$ with $g(u) < f(v) < g(w)$.

Likewise, the image of an edge $e = (u,v) \in E_i^\X$ in $\Y'$ is a monotone path which begins at $\phi^V(u)$ and ends at $\phi^V(v)$.  
In terms of the combinatorial structure, this path is a sequence of  edges $e_1,e_2,\cdots, e_k$ such that the top vertex of $e_i$ is the bottom vertex of $e_j$.

\subsection{The canonical map from an $\R$-graph to its smoothing}

Let $(\X,f)$ be an $\R$-graph and let $(\Y,g) = (\X_\e, f_\e)$ be its smoothing. In order to recognize interleavings we need access to the canonical map $\zeta = \zeta^\e_f: (\X,f) \to (\Y,g)$. We show how to obtain this in terms of the alternate description of maps described in the preceding subsection.

We do this during the course of constructing $(\Y,g)$. During the sweep, in addition to stopping at critical values $\{b_i\}$ of the function~$g$ we will also stop at the critical values $\{a_i\}$ of~$f$. This way, we can determine $\zeta(v)$ by determining the component of $H$ containing $v$, and finding the edge or vertex of $\Y$ to which it is associated. (It is usually an edge, unless by chance we have some $|a_i - a_j| = \e$.) We retain this information for each vertex $v \in V(\X)$.
For each edge $e \in E(\X)$, we must record the corresponding path $\zeta(e)$. This we do by maintaining a list with the edge. At each update of $H$, we find the representative of the component of $H$ containing $e \in V(H)$.
Thus, when the algorithm ends, the map $\zeta$ has been completely determined.

\subsection{Complexity of the Reeb interleaving distance}

We have not given a method for calculating the Reeb graph interleaving distance, because it is not easy. Let us first take the infimum out of consideration, by selecting $\e \geq 0$. This leads to the question in the next result.

\begin{prop}
\label{prop:NP}
``Can $(\X,f)$ and $(\Y,g)$ be $\e$-interleaved?'' is in $\textbf{NP}$.
\end{prop}

\begin{proof}
We need to identify possible certificates and a polynomial-time verification that those certificates guarantee an $\e$-interleaving. We begin by constructing $\UU_\e f$, $\UU_{2\e}f$, $\UU_\e g$, $\UU_{2\e}g$ and the morphisms $\zeta^{2\e}_f$ and $\zeta^{2\e}_g$.
A candidate certificate is a pair of maps $\alpha: f \to \UU_\e g$ and $\beta: g \to \UU_\e f$.
We calculate $\alpha^\e_{2\e}$ and $\beta^\e_{2\e}$ and the composites $\alpha^\e_{2\e} \circ \beta$ and $\beta^\e_{2\e} \circ \alpha$ and return the answer ``yes'' if the equations
\[
\alpha^\e_{2\e} \circ \beta = \zeta^{2\e}_g
\quad
\text{and}
\quad
\beta^\e_{2\e} \circ \alpha = \zeta^{2\e}_f
\]
both hold. By \eqref{eq:interleaving-graph}, the correct answer is ``yes'' if and only if there exists such a pair $\alpha, \beta$ satisfying these tests.
All the constructions and checks can be done in polynomial time, so the problem is in $\NP$.
 \end{proof}

This theorem leaves open-ended the difficulties of finding such an interleaving.  When $\e = 0$ we have the following obstruction.

\begin{prop}
``Is there a $0$-interleaving between $(\X,f)$ and $(\Y,g)$?'' is graph-isomorphism hard.
\end{prop}

We thank Tamal Dey and Jeff Erickson for the following argument.

\begin{proof}
We find a reduction from the graph isomorphism problem to the $0$-interleaving problem. Let $G$ and $H$ be two finite graphs for which we wish to test isomorphism. Let us vertices in $V(G)$ and $V(H)$ as basepoints. Is there an isomorphism $G \cong H$ which sends basepoint to basepoint? 

To answer this question, convert each graph to an $\R$-graph by using the distance from basepoint as a function. Some edges may need to be split in two (if their two vertices are equidistant from the basepoint, meaning that the distance function will increase towards the middle of the edge); otherwise the vertices and edges remain the same.
Then these two $\R$-graphs are $0$-interleaved if and only if they are isomorphic as $\R$-graphs (by Proposition~\ref{Prop:IsoReeb}), if and only if $G \cong H$ preserving basepoints.

To test whether $G,H$ are isomorphic without reference to basepoints, it suffices to repeat this test keeping the basepoint of~$G$ fixed and varying the basepoint of~$H$ over all possible vertices. If the test fails every time, then $G, H$ are not isomorphic to each other. If it succeeds even once, then they are isomorphic.
In this way, a solution to the $0$-interleaving problem for $\R$-graphs gives a solution to the graph-isomorphism problem.
 \end{proof}

By Proposition~\ref{prop:rdist=0}, it follows equivalently that ``Is the interleaving distance between $(\X,f)$ and $(\Y,g)$ equal to zero?'' is graph-isomorphism hard.

\section{Discussion}
\label{sec:discussion}

There is a sense in which the development in this paper is self-annihilating: all of our cosheaf constructions are in the end realized geometrically. At least, this is true in the constructible case. With that in mind, there are perhaps two main motivations for our work:

\begin{itemize}
\item
The correspondence between constructible cosheaves and $\R$-graphs allows us to transfer ideas from one realm to the other. For instance, the Reeb cosheaf distance occurs very naturally in the context of persistence but its geometric equivalent is not an obvious construction.

\item
The realm of cosheaves is broader that the realm of $\R$-graphs, since we can work quite easily with non-constructible cosheaves whereas $\R$-graphs are necessarily constructible. We don't claim any immediate applications for this greater generality, but it is good to know that it is available.
\end{itemize}

\noindent
A pleasant consequence of this thinking is that the Reeb graph emerges as yet another instance of topological persistence, standing alongside the persistence diagram and the dendrogram (or join-tree) as a persistent invariant. It shares with those invariants an `interleaving' strategy for defining a distance, and an easily accessed stability theorem.

\bigskip
Here are a few closing remarks.

\paragraph{Persistence.}
Our approach was inspired by the paper of Morozov, Beketayev and Weber~\cite{Morozov2013} who defined interleavings of join-trees (called `merge trees' in that paper) in a geometric way. This ties in nicely with Bubenik and Scott's approach to persistence~\cite{Bubenik2014}: their `generalized persistence modules' are functors from the real line to an arbitrary category, and interleavings can be defined in terms of translations of the real line. By regarding a join-tree as a $\Set$-valued persistence module, the two approaches lead to the same interleaving distance and the same stability theorem.

More generally, one may define interleaving distances for persistence modules over an arbitrary poset~\cite{Bubenik_dS_S_2014}. In our present work, we regard Reeb graphs as $\Set$-valued functors on the poset $\Int$ of open intervals in the real line. 
Because these functors are cosheaves, we get a tight relationship between the geometric and the category theoretic points of view \cite{Funk1995,Treumann09,Woolf_2009}. That said, the cosheaf condition is not at all needed to define the interleaving distance or to obtain the stability theorem.
Functors on $\Int$ analogous to the Reeb cosheaf may be obtained by replacing $\pi_0$ with, for example, $\pi_k$ or $\mathrm{H}_k$. Since the result is usually not a cosheaf, the question remains how to manage these objects.

\paragraph{Higher dimensions.}
Reeb graphs easily generalize to Reeb spaces. Given a continuous map $f : \X \to \mathbb{M}$, we say two points $x, y \in \X$ are \emph{equivalent} if $f(x) = f(y) = p$ and if $x$ and $y$
lie on the same path-component of the fiber $f^{-1}(p)$. The \emph{Reeb space} is the resulting quotient space together with its induced map to~$\mathbb{M}$.

The structure of constructible Reeb spaces is rich and much is known. For example, suppose~$\mathbb{M}$ is a manifold and suppose the map $f$ is stratifiable 
with respect to a stratification $\mathcal{S}$ of~$\mathbb{M}$.
Then its Reeb space is also stratifiable with respect to~$\mathcal{S}$. Using an argument similar to the one presented in this paper, one can show that an $\mathcal{S}$-constructible Reeb space is equivalent to a finite-set-valued $\mathcal{S}$-constructible cosheaf over~$\mathbb{M}$.

Going further, Robert MacPherson observed (in unpublished work) that $\mathcal{S}$-constructible cosheaves can be completely classified in the following way. 
He constructed a category $\mathbf{Ent}(\mathbb{M}, \mathcal{S})$, called the \emph{entrance path category}, whose objects are the points of~$\mathbb{M}$ and whose morphisms $p \to q$ are homotopy classes of \emph{entrance paths} from $p$ to~$q$. These are paths $\gamma: [0,1] \to \mathbb{M}$ from $p$ to~$q$ which are allowed to move only from higher- to lower-dimensional strata. More precisely, whenever $s \leq t$, the stratum containing $\gamma(t)$ is contained in in the closure of the stratum containing $\gamma(s)$. Homotopies are required to remain within this class of paths, with fixed endpoints.
MacPherson's result is that $\mathcal{S}$-constructible cosheaves are essentially the same thing as functors on $\mathbf{Ent}(\mathbb{M}, \mathcal{S})$.
A special case is seen in Propositions \ref{prop:cshc-data} and~\ref{prop:cshc-morphism}, where constructible cosheaves with critical set~$S$ and their morphisms are shown to be equivalent to functors on the zigzag diagram~\eqref{eq:entrance} and natural transformations between them. Here, diagram~\eqref{eq:entrance} is interpreted as a category, and this category is equivalent to $\mathbf{Ent}(\mathbb{R},S)$.
We recommend the papers of Treumann~\cite{Treumann09}, Woolf~\cite{Woolf_2009} and Curry~\cite{Curry2014} as further reading on these topics.

Edelsbrunner, Harer and Patel~\cite{reeb_spaces} gave the first algorithm to compute the Reeb space of a piecewise linear
map from a simplicial complex to the plane. Their algorithm computes the coarsest stratification of the Reeb space. There is much room for improvement in its running time.

\paragraph{Non-Constructibility.}
In our treatment, we assign special importance to the constructible objects in the categories of $\R$-spaces and cosheaves. We do this partly for the sake of the equivalence theorem: we do not know how to interpret non-constructible cosheaves in a cleanly geometric way. And the loss of generality is not too great, because many regular situations (Morse functions on a compact manifold, definable functions on a compact semialgebraic set, etc) supply us with constructible $\R$-spaces.
Our definition is perhaps too broad: our constructible $\R$-spaces are built out of spaces which are locally path-connected, which is enough for our purposes. On the other hand, if we want to use higher-dimensional functors $\pi_k$ or $\mathrm{H}_k$ we ought to request local contractibility.  The usual well-behaved examples have this stronger property.

On the other hand, much of the theory works just as well for non-constructible cosheaves; so it would be a pity to rule them out altogether. Is there a corresponding equivalence theorem? 
One promising possibility is to define $\mathbf{Stone}$-valued rather than $\Set$-valued Reeb cosheaves. The connected components of a topological space can be regarded not merely as a set, but as a quotient space. If the original space is compact, then the quotient is both compact and totally disconnected, and is called a \emph{Stone space}. Our suggestion is to consider $\R$-spaces $(\X,f)$ where $\X$ is locally compact and Hausdorff, and $f$ is proper. Then $\overline{\pi}_0 f^{-1}[a,b]$ is a Stone space for every compact interval $[a,b]$. In this setting, Example~\ref{ex:non-cosheaf} is no longer a point of failure of the cosheaf condition for $\overline{\pi}_0$: the colimit described there becomes a non-Hausdorff two-point space when evaluated in the topological category, and therefore becomes the desired one-point space when evaluated in a category of Hausdorff spaces such as $\mathbf{Stone}$. All of this can be done in the category of (not necessarily compact) totally disconnected spaces, but we suspect that Stone spaces may be more fruitful; in part because their category is known to be dual to the category of Boolean algebras. A necessary requirement is to work with cosheaves over the `site' of compact intervals rather than the more usual open intervals.

\paragraph{Simplification versus smoothing.}
We have deliberately chosen not to describe our topological smoothing algorithm for Reeb graphs as topological `simplification', a term that is widely used elsewhere~\cite{Pascucci2007a}. The distinction between simplification and smoothing is as follows. Simplification seeks to remove unnecessary loops and edges; and the methods that exist carry this out in an \emph{ad hoc} local fashion by, for instance, collapsing small loops. Larger features are expected to retain their metric properties; this makes simplification a difficult problem. Smoothing, on the other hand, modifies the graph and its function on a global scale. There are simplifying effects, such as the disappearance of small `noisy' loops as the smoothing parameter is increased; but there are also global side-effects such as the steady divergence of the maximum and minimum values of the function. We haven't in this work attempted to identify what specific geometric simplifications take place as a result of smoothing. Some such questions are addressed in~\cite{Bauer2014a}.

\paragraph{Computational complexity.}
While the computation of the smoothed Reeb graph is algorithmically reasonable, we are still working to find sensible strategies for computing or estimating the interleaving distance. We can narrow our concerns to specific values of~$\e$, using binary splitting to converge towards the true value. For a fixed~$\e$, the existence of an $\e$-interleaving we have seen to be difficult.
The obvious na\"{i}ve algorithm would follow Proposition~\ref{prop:NP}, testing all possible morphism pairs $\alpha: f \to \UU_eg$ and $\beta: g \to \UU_\e f$ for their validity as certificates. This will certainly have worst-case exponential running time.

That said, we think there are reasons to be hopeful. We know that $\epsilon = 0$ is difficult, but on the other hand when $\epsilon = \text{large}$ all the internal structure disappears and the question becomes trivial. What happens along the way?
There may be classes of $\R$-graphs which can be compared efficiently; and there may be worst-case bad algorithms which behave well in practice most of the time.
We expect that there exist easily-computed invariants that obstruct the existence of $\e$-interleavings; a trivial example is given in Proposition~\ref{Prop:IsoReeb}, and there should be many more. There may be ways of selecting certificates $\alpha, \beta$ intelligently, taking the required relations into account. It seems to us that there are plenty of open questions here.

It has also been shown~\cite{Bauer2014a} that the interleaving distance is tightly related to another metric on Reeb graphs, the functional distortion distance~\cite{Bauer2014}; in particular, the two metrics are strongly equivalent. It is possible that this relationship will allow for an interplay of methods for computation between the two methodologies; however, it is more likely that this relationship will be most directly useful for passing any NP-hardness results between the two metrics.

\paragraph{}
We finish by acknowledging our other goal in writing this paper: to support the introduction of category theory and sheaves to the computational geometry community. There have been many other such efforts which we are proud to stand alongside. To name a few: Robert Ghrist, Justin Curry, Sanjeevi Krishnan, Michael Robinson and Vidit Nanda have developed many new applications of sheaves, cosheaves and their generalizations; and Peter Bubenik and Jonathan Scott have introduced category-theoretic thinking to the study of persistence.
We hope that our direct presentation of cosheaf methods will help  promote their work. Sometimes a Reeb graph is just a Reeb graph;  but very often, secretly, there is a cosheaf waiting to get out.


\section*{Acknowledgements}
The work described in this article is a result of a collaboration made possible by the Institute for Mathematics and its Applications (IMA), University of Minnesota. It was carried out while the authors were  at the IMA during the annual program 2013--14, on {\it Scientific and Engineering Applications of Algebraic Topology}. We thank the IMA staff for their outstanding support and help throughout the year, and we thank the organizing committee of the thematic year for putting together an excellent program. The IMA is funded by the National Science Foundation.

EM and AP have been supported by IMA Postdoctoral Fellowships during this project.
VdS thanks his home institution, Pomona College, for granting him a sabbatical leave of absence in 2013--14. The sabbatical was hosted by the IMA, and the second semester of leave was supported by the Simons Foundation (grant \#267571). 

The authors gladly thank Ulrich Bauer, Justin Curry, Tamal Dey, Jeff Erickson, Robert MacPherson, Dmitry Morozov, Sara Kali\v{s}nik, Mikael Vejdemo-Johansson, Yusu Wang and John Wilmes for many helpful discussions during the course of this work.


\bibliographystyle{plain}
\bibliography{ReebGraphs}

\end{document}